\documentclass[11pt]{article}
\usepackage{graphicx} 
\usepackage{authblk}
\usepackage{amsmath,amssymb,mathrsfs}
\usepackage{amsthm}
\usepackage{xcolor}
\usepackage{float}
\usepackage[margin=0.8in]{geometry}
\usepackage{hyperref}
\usepackage{subcaption}
\newtheorem{theorem}{Theorem}
\newtheorem{lemma}{Lemma}
\newtheorem{corollary}{Corollary}
\newtheorem{proposition}{Proposition}
\theoremstyle{definition}
\newtheorem{definition}{Definition}
\theoremstyle{remark}

\theoremstyle{plain}
\title{Reverse quantum state diffusion from differential geometry}
\author[1]{Dinh-Long Vu\thanks{long-vu@nus.edu.sg}}
\author[1]{Patrick Rebentrost\thanks{cqtfpr@nus.edu.sg}}
\affil[1]{Centre for Quantum Technologies, National University of Singapore, Singapore, Singapore}
\date{}
\begin{document}
\maketitle
\begin{abstract}
Quantum state diffusion unravels a Lindblad master equation into stochastic trajectories of pure states. The ensemble of trajectories is described by a probability density on the manifold of pure states, which contains much more information than the density matrix, its first moment. We derive the time reversal of this diffusion. Writing the reversal of a general diffusion on a manifold in Stratonovich form, and passing through the manifold of pure states, we obtain a backward stochastic differential equation for the state vector, valid for any Hilbert space dimension and any Lindbladian. The backward stochastic equation provides a physical, nonlinear unravelling of the inverse Lindbladian that holds only for the ensemble from which it was constructed. Following score-based generative models, we then use the backward equation to generate a new ensemble close to the original. For the depolarizing channel the unravelling is a Brownian motion on complex projective space, so the uniform distribution of pure states is its stationary law in any dimension. Starting the backward equation from that prior, we bound the error of the resulting distribution compared to the original one in different settings: exact and learned scores, continuous and discrete time.
\end{abstract}

\section{Introduction}\label{sec: intro}
The dynamics of an open quantum system weakly coupled to a Markovian environment is described by a master equation for its density matrix $\rho$. The most general equation that generates a completely positive and trace preserving (CPTP) map has the Lindblad, or Gorini-Kossakowski-Sudarshan-Lindblad, form \cite{lindblad1976generators, gorini1976completely, breuer2002theory}. The same dynamics can be \emph{unravelled} into stochastic trajectories of pure states, such that the average of the projector  over the trajectories reproduces the density matrix at all times \cite{carmichael1993open, dalibard1992wave, breuer2002theory, wiseman2010quantum}. Physically, a trajectory describes the state of the system conditioned on the record of a continuous measurement of its environment \cite{wiseman2010quantum}. Among the unravellings, quantum state diffusion  \cite{gisin1992quantum, gisin1997quantumstatediffusionfoundations} describes the state vector by a continuous, norm-preserving stochastic differential equation (SDE), often called Gisin's equation; it corresponds to heterodyne detection of the environment \cite{wiseman2010quantum}. 

The questions that started this work are the following: can Gisin's equation be reversed, and would this reversal be compatible with the inverse Lindblad evolution, which is not a CPTP map?

The reversal of a classical SDE involves the gradient of the log probability density, the so-called score function \cite{nelson1967dynamical, anderson1982reverse, haussmann1986time}. In quantum mechanics, the density matrix is often thought of as the quantum analog of a classical probability distribution. However, we find this analogy misleading in the context of quantum state diffusion. Gisin's equation is a classical diffusion on the manifold of pure states and its probability distribution obeys a classical Fokker-Planck equation on this manifold.  The density matrix  is only the first moment of this distribution and infinitely many distributions share the same density matrix. For instance, the maximally mixed qubit state can either be seen as the average of two orthogonal states or of the uniform distribution on the Bloch sphere. Higher order statistics distinguish these ensembles but are invisible to the master equation. There is thus a clear conceptual and technical distinction between the distribution of pure states and the density matrix.

Applying the general machinery of SDE on  Riemannian manifolds to the complex projective space, we obtain a backward SDE for the state vector (Theorem \ref{thm: backward qsd}). It has the same noise operators as Gisin's equation, which are the jump operators shifted by their mean. Its drift consists of the forward drift with the opposite sign, a geometric term native to the complex projective space (Lemma \ref{lem: projective divergence}), and the scores, which are the derivative of the log density along the noise directions. In taking the average over the ensemble of pure states, we show that the density matrix of the backward ensemble obeys the inverse Lindblad evolution (Corollary \ref{corr: reverse lindblad}). Although the inverse evolution is not a CPTP map, the backward SDE always produces physical states: while Gisin's equation unravels $\mathcal{L}$ for any ensemble, the backward SDE only unravels $-\mathcal{L}$ for the ensemble from which it was built. This settles the questions that we initially set out to answer.

Moving beyond the initial goal, we ask whether this backward SDE can be used to generate new pure states from the initial ensemble, just like in classical score-based generative models \cite{sohldickstein2015deep, ho2020denoising, song2021scorebased}. There, data are gradually turned into noise by a forward diffusion, and new data are generated by running the reverse diffusion with a learned score.  On the quantum side, we find that for the depolarizing channel, Gisin's SDE defines a Brownian motion on the complex projective space (Proposition \ref{prop: forward backward depol}). The Fubini-Study measure is then the reversible stationary law, and it plays the role that the standard Gaussian plays in flat space: a universal prior, available in any dimension.  Initializing the backward SDE from this universal prior, we bound the error of the resulting distribution compared to the original one under various settings: exact scores and continuous time (Proposition \ref{prop: prior error}), learned scores and continuous time (Proposition \ref{prop: girsanov}) and finally the most realistic one: learned scores and discrete time (Proposition \ref{prop: discretization}). We test the bounds in dimensions up to $d=8$.

The paper is organized as follows. Section \ref{sec: qsd} collects relevant facts about quantum state diffusion, in particular Gisin's equation in It\^o and Stratonovich form. Section \ref{sec: reversal} derives the time reversal of a diffusion on a Riemannian manifold, in complex vector fields notation. Section \ref{sec: backward state vector} applies it to quantum state diffusion and derives the backward SDE for the state vector. Section \ref{sec: prior} identifies the unravelling of the depolarizing channel as a Brownian motion and studies the backward SDE started from its stationary measure, with an exact score, with a learned score and with a finite time step. Section \ref{sec: examples} contains the numerical tests, and Section \ref{sec: discussion} concludes. The proofs are collected in Appendices \ref{app: projective divergence}-\ref{app: step sizes}, Appendix \ref{app: numerics} describes the numerical methods, and Appendix \ref{app: bgg} relates Theorem \ref{thm: backward qsd} to the reverse SDE of \cite{garrahan2026}.

\paragraph{Comparison with the work by Bompais, Gu\c{t}\u{a} and Garrahan.} After the manuscript was written, we learned of the paper \cite{garrahan2026}. This paper addresses the same questions that we try to answer here, namely the reversal of Gisin's SDE and the task of generating a new ensemble from the original one, so it is important to clarify the similarities and differences between the two papers. First and foremost, regarding the reversal of Gisin's SDE, their main theorem and our Theorem \ref{thm: backward qsd} describe the same reverse dynamics: once the divergence term in their result is read as a distribution, it can be shown to coincide with the geometric and score terms of our backward SDE (Appendix \ref{app: bgg}). We thus acknowledge that \cite{garrahan2026} is the first work to derive the reverse quantum state diffusion equation. The two papers use completely different approaches to derive this result however, which affects how the result can be used, in particular for the score-based application. 

At a high level, \cite{garrahan2026} treats relevant objects, i.e.,  quantum states, jump operators and the Hamiltonian, explicitly as matrices and they use matrix calculus to obtain the backward generator in It\^o convention. Since they work with the Lebesgue measure on the real vector space of Hermitian matrices, with respect to which the law of pure states (which live on a submanifold of measure zero) is singular, several steps in their derivation are formal. In contrast, we work directly on the complex projective space of pure states with the Fubini-Study measure. We use Stratonovich convention, which makes the drift and diffusion coefficients tangent vector fields, so the entire dynamics stays on this manifold. Thus, their formulation is embedded while ours is intrinsic. Not only does this distinction make our derivation more rigorous, it also has two important practical consequences.

First, the complex projective manifold of pure states is highly symmetric. Our formulation makes symmetries explicit, while treating pure states as matrices hides the rich underlying geometric structure. Take for instance the depolarizing channel. It is observed in \cite{garrahan2026} that, for a single qubit, the correction to the feedback Hamiltonian only involves the score term, while the divergence of the diffusion tensor cancels with the dissipator. This is confirmed by explicit calculations involving Pauli matrices and the origin of this cancellation is not apparent. In our formalism, this cancellation is 
a direct consequence of the fact that the SDE generator in this case is the Laplace-Beltrami operator, i.e., the generalization of the Laplace operator in flat space to Riemannian manifolds. This is the reason why the depolarizing channel unravels to a Brownian motion, in any dimension. Moreover, this gives us access to the full spectrum of the generator, since the spectrum of the Laplace-Beltrami operator on the complex projective space is known. The spectral gap of the generator governs the convergence rate of the diffusion, while higher eigenvalues control the sub-leading corrections (Proposition \ref{prop: prior error}). 

Second, while the direct matrix calculus approach can establish certain equalities (like the cancellation above), inequalities are out of reach in most cases. In particular, no error bound for the distance between the generated ensemble and the original one is provided in \cite{garrahan2026}. Our paper provides a bound in all settings: exact and learned scores, continuous and discrete time. In the most practical case of learned scores and discrete time, geometry gives us two important properties. First, by a Noether-type argument, the derivative of the log density along a symmetry generator is conserved in expectation along the reverse process. As a result, the coefficients of the backward Hamiltonian are martingales.  Second, by orthogonality of martingale increments, the variation of these coefficients within a time interval is bounded by the increment of the Fisher information along the noise fields. The Fubini-Study metric on the complex projective space has a  positive Ricci curvature so the Li-Yau inequality \cite{li1986parabolic} can be used to bound the Fisher information. These properties lead to Proposition \ref{prop: discretization} and Corollary \ref{cor: step sizes}.

\section{Quantum state diffusion}\label{sec: qsd}
This section presents basic facts about quantum state diffusion which will be used in the rest of the manuscript. For It\^{o} and Stratonovich conventions for stochastic integrals, see \cite{oksendal2003stochastic, karatzas1991brownian}.

Consider a $d$-dimensional quantum system whose density matrix obeys the Lindblad equation \cite{lindblad1976generators, gorini1976completely}
\begin{align}
    \frac{d\rho}{dt} = \mathcal{L}(\rho) := -i[H,\rho] + \sum_{k} \Big(L_k\rho L_k^\dagger -\frac{1}{2}\big\lbrace L_k^\dagger L_k,\rho\big\rbrace\Big),
    \label{eq: lindblad}
\end{align}
with Hamiltonian $H$ and jump operators $L_k$. Gisin's SDE for the state vector reads, in It\^{o} form \cite{gisin1992quantum, gisin1997quantumstatediffusionfoundations},
\begin{align}
    |d\psi\rangle = A|\psi\rangle\, dt + \sum_kB_k|\psi\rangle\, d\xi_k,
    \label{eq: Gisin SDE}
\end{align}
with the state-dependent operators
\begin{align}
    A := -iH -\frac{1}{2}\sum_kL_k^\dagger L_k  -\frac{1}{2}\sum_k|l_k|^2  +\sum_kl_k^*L_k, \qquad B_k := L_k-l_k, \qquad l_k := \langle \psi|L_k|\psi\rangle .
    \label{eq: gisin coefficients}
\end{align}
Here $l_k$ is the expectation value of $L_k$ in the current state. The equation is nonlinear in $|\psi\rangle$, because $A$ and $B$ depend on $|\psi\rangle$ through $l_k$. The complex Wiener increments $d\xi_k$ are independent: $d\xi_j\,d\xi_k^* = \delta_{jk}\,dt$ and $d\xi_j\,d\xi_k = 0$. One can check that  $d\langle\psi|\psi\rangle = 0$: the state vector remains normalized. Let $\psi := |\psi\rangle\langle\psi|$ be the projector onto the state; we use the same letter for the state vector and the projector when no confusion can arise. By It\^{o}'s product rule,
\begin{align}
    d\psi = |d\psi\rangle\langle\psi| + |\psi\rangle\langle d\psi| + |d\psi\rangle\langle d\psi| .
\end{align}
After some simple manipulations, we have the It\^o SDE for the projector
\begin{align}
    d\psi = \mathcal{L}(\psi)\,dt + \sum _k \Big(B_k\psi\, d\xi_k + \psi B^\dagger_k d\xi^*_k \Big).
    \label{eq: projector sde}
\end{align}
This is more illuminating than  the SDE \eqref{eq: Gisin SDE} for the state vector as the drift is the Lindbladian itself. Taking the expectation of \eqref{eq: projector sde} recovers the Lindblad equation \eqref{eq: lindblad}.  For an operator $O$, we write $\langle O\rangle := \mathrm{tr}(O\psi) = \langle\psi|O|\psi\rangle$ for its expectation value in the pure state $\psi$.  Its average over the ensemble is $\mathbb{E}[\langle O\rangle] = \mathrm{tr}(O\rho)$, which we always write in the latter form to avoid confusion. Multiplying \eqref{eq: projector sde} by $O$ and taking the trace gives the SDE for $\langle O\rangle$
\begin{align}
    d\langle O\rangle = \langle\mathcal{L}^\dagger(O)\rangle\,dt + \sum_k\Big(\langle\psi|O(L_k-l_k)|\psi\rangle\, d\xi_k + \langle\psi|(L_k^\dagger-l_k^*)O|\psi\rangle\, d\xi_k^*\Big),
    \label{eq: observable sde}
\end{align}
where the Heisenberg-picture generator $\mathcal{L}^\dagger$ is given by
\begin{align}
    \mathcal{L}^\dagger[O] = i[H,O] + \sum_k\Big(L_k^\dagger OL_k - \frac12\lbrace L_k^\dagger L_k,O\rbrace\Big) .
\end{align}
It is dual to $\mathcal{L}$ in the sense that $\mathrm{tr}(O\,\mathcal{L}[\psi]) = \mathrm{tr}(\mathcal{L}^\dagger[O]\,\psi)$ for all $\psi$.

While It\^o calculus is more suitable for \textit{killing the noise} when needed, Stratonovich calculus obeys the ordinary chain rule, which makes it the preferred language for diffusions on manifolds \cite{hsu2002stochastic}. To convert between the two conventions, a correction needs to be added to the drift.  In our setting, this correction is $\bar l/2$ with
\begin{align}
    \bar l := \sum_k|l_k|^2 -\sum_k  \langle\psi|L_k^\dagger L_k|\psi\rangle = -\sum_k\langle\psi|B^\dagger_k B_k|\psi\rangle .\label{eq: l bar}
\end{align}
The Stratonovich form of Gisin's SDE for the state vector thus reads
\begin{align}
    |d\psi\rangle = \Big(A-\frac12\bar{l}\Big)|\psi\rangle\, dt + \sum_kB_k|\psi\rangle\circ d\xi_k .
    \label{eq: gisin stratonovich}
\end{align}
An equivalent form of equation \eqref{eq: gisin stratonovich} appears in the work of Di\'{o}si, Gisin and Strunz \cite{diosi1998nonmarkovian, Strunz_1999}. For the projector, we have the following
\begin{align}
    d\psi = \Big[\Big(A-\frac12\bar l\Big)\psi + \psi\Big(A-\frac12\bar l\Big)^\dagger\Big]\,dt +\sum_k \Big( B_k\psi\circ d\xi_k + \psi B_k^\dagger\circ d\xi^*_k\Big).
    \label{eq: gisin projector stratonovich}
\end{align}

\section{Time reversal of a diffusion on a Riemannian manifold}\label{sec: reversal}
In this section we present the general formalism of reverse diffusion on a Riemannian manifold. Although the derivation is known \cite{ debortoli2022riemannian, huang2022riemanniandiffusionmodels}, we quickly sweep through it mainly to introduce the notations in a natural order. Nothing here is specific to quantum mechanics except for the complex vector field formulation. The basic notions of differential geometry that we use, vector fields, divergence and integration by parts can be found in \cite{lee2013smooth, lee2018riemannian}.  

The logic is exactly the same as in flat space, the only difference is that we use the volume element induced from the Riemannian metric for integration and we need to express relevant quantities in a covariant form. The flow starts with the forward diffusion generator, which can be read off directly from the drift and the diffusion coefficients in It\^o form. This generator acts on functions, which is the Heisenberg picture in a quantum physicist's language. Equivalently, its adjoint acts on probability distributions in the Schr\"odinger picture and this is the Fokker-Planck equation. The time reversal of the Fokker-Planck equation reverses the sign of the adjoint generator, which can be rewritten as a diffusion generator with the old diffusion coefficients but with a new, shifted drift containing the score.  We then convert the entire pipeline to the Stratonovich convention.

\paragraph{It\^o form.} Let $X_t$ be a diffusion on a compact manifold $\mathcal{P}$ of dimension $N$ without boundary; in our setting $\mathcal{P}$ will be the manifold of pure states. In local coordinates $x^1,\dots,x^N$ we write the It\^{o} SDE
\begin{align}
    dX^i = b^i(X)\,dt + \sum_{k=1}^M\sigma^i_k(X)\, dW^k,\qquad D^{ij} := \sum_{k=1}^M \sigma^i_k\sigma^j_k ,
\end{align}
driven by $M$ independent real Wiener processes $W^k$. The matrix $D$ is symmetric and positive semi-definite. For a smooth function $f$ on $\mathcal{P}$, the mean rate of change of $f(X_t)$ is given by
\begin{align}
    \mathbb{E}\big[df(X_t)\,\big|\,X_t=x\big] = (\mathcal{W}f)(x)\,dt,\qquad
    \mathcal{W}f := \sum_{i=1}^N b^i\,\partial_i f + \frac{1}{2}\sum_{i,j=1}^N D^{ij}\,\partial_i\partial_j f .
    \label{eq: generator}
\end{align}
The second-order differential operator $\mathcal{W}$ is the \emph{generator} of the diffusion \cite{oksendal2003stochastic, Bakry_2014}. 

Let $dm = \sqrt{\mathbf{g}}\,dx^1\cdots dx^N$ be a reference volume element, for instance the Riemannian volume form with $\mathbf{g}:=\det g$ being  the determinant of the metric tensor, and let $p_t$ be the density of $X_t$ with respect to $dm$. The change in the expectation value of a function can be expressed either through the change of the function or the change of the probability distribution, i.e., 
\begin{align*}
    \int f\,\partial_tp_t\,dm = \int(\mathcal{W}f)\,p_t\,dm .
\end{align*}
The right-hand side can be integrated by parts to give $\int(\mathcal{W}f)\,p\,dm = \int f\,(\mathcal{W}^*p)\,dm$ with
\begin{align}
    \mathcal{W}^*p = -\frac{1}{\sqrt{\mathbf{g}}}\sum_i\partial_i\big(\sqrt{\mathbf{g}}\, b^i p\big) + \frac{1}{2}\frac{1}{\sqrt {\mathbf{g} }}\sum_{i,j}\partial_i\partial_j\big(\sqrt{\mathbf{g}}\, D^{ij}p\big).
\end{align}
Since this holds for all functions, $\partial_t p = \mathcal{W}^*p$. This is the Fokker-Planck equation in It\^{o} form. 

Fix a final time $T$ and let $Y_\tau := X_{T-\tau}$, $0\le\tau\le T$, be the reversed process. Its density at reverse time $\tau$ is $p_{T-\tau}$. By the Fokker-Planck equation of the forward process, and since $\partial_\tau = -\partial_t$,
\begin{align}
    \partial_\tau p_{T-\tau} = -\mathcal{W}^*p ,
\end{align}
where $p $ is implicitly understood as $ p_{T-\tau}$. The diffusion term has the wrong sign for a Fokker-Planck equation and can be \textit{absorbed} into the drift, assuming that $p_t$ is smooth and positive for $t>0$
\begin{align}
    \partial_\tau p_{T-\tau} = -\frac{1}{\sqrt {\mathbf{g} }}\sum_i\partial_i\big(\sqrt {\mathbf{g} }\, \tilde b^ip\big) + \frac12\frac{1}{\sqrt {\mathbf{g} }}\sum_{i,j}\partial_i\partial_j\big(\sqrt {\mathbf{g}}\, D^{ij}p\big),\qquad \tilde b^i = -b^i + \frac{1}{p\sqrt {\mathbf{g}}}\sum_j\partial_j\big(\sqrt{\mathbf{g} }\,p\,D^{ij}\big). \label{eq: reverse drift general}
\end{align}

Now that the reverse Fokker-Planck equation has been written in the standard form \eqref{eq: reverse drift general}, the reverse generator can be read off directly from it, as in \eqref{eq: generator}
\begin{align}
    \tilde{\mathcal{W}}_tf = -\mathcal{W}f + \frac{1}{p\sqrt {\mathbf{g}}}\sum_{j}\partial_j\Big(\sqrt{\mathbf{g}} \, p\sum_i D^{ji}\,\partial_i f\Big).
\end{align}
For a vector field $V=\sum_jV^j\partial_j$ the divergence with respect to $dm$ is $\mathrm{div}\,V = \frac{1}{\sqrt{\mathbf{g}}}\sum_j\partial_j(\sqrt{\mathbf{g}}\,V^j)$. Contracting $D$ with the differential of $f$ defines a vector field that we denote by $D\nabla f$,
\begin{align}
    (D\nabla f)^j := \sum_i D^{ji}\,\partial_i f .
\end{align}
We can therefore write the reverse generator in a covariant form
\begin{align}
    \tilde{\mathcal{W}}_t f = -\mathcal{W} f + \frac{1}{p_t}\,\mathrm{div}\big(p_t\, D\nabla f\big).
    \label{eq: nelson}
\end{align}
\paragraph{Stratonovich form.} Consider now the Stratonovich SDE
\begin{align}
    dX^i = V_0^i(X)\,dt + \sum_{a=1}^M V_a^i(X)\circ dW^a .
\end{align}
Under a change of coordinates,  the coefficients $V_a^i$ transform like the components of a vector field. It is therefore natural to regard $V_0,V_1,\dots,V_M$ as vector fields on $\mathcal{P}$, acting on functions as derivations, $V_af := \sum_iV_a^i\,\partial_if$. Writing \eqref{eq: generator} in terms of $V_a$, the forward generator takes the compact form
\begin{align}
    \mathcal{W}f = V_0f + \frac12\sum_a V_a(V_af).
    \label{eq: hormander}
\end{align}
Substituting this equation into \eqref{eq: nelson}, using the fact that $D\nabla f = \sum_ a (V_a f)V_a$, we can write the reverse generator in the same form
\begin{align}
    \tilde{\mathcal{W}}_tf = \Big[-V_0 + \sum_a\mathrm{div}_{p_t}(V_a)\,V_a\Big]f + \frac12\sum_a V_a(V_af),
\end{align}
where we have introduced the divergence with respect to the measure $p\,dm$,
\begin{align}
    \mathrm{div}_{p}(V) := \frac{1}{p}\mathrm{div}(pV) = \mathrm{div}\,V + V(\log p).\label{eq: div p}
\end{align}
The reverse Stratonovich SDE thus has the same noise fields $V_a$ and the drift
\begin{align}
    \tilde{V}_0 = -V_0 + \sum_a \mathrm{div}_{p_t}(V_a)\, V_a.
    \label{eq: stratonovich reversal}
\end{align}
The correction splits into a geometric part $\sum_a(\mathrm{div}\,V_a)V_a$, which does not depend on $p_t$, and the score part $\sum_aV_a(\log p_t)V_a = D\nabla\log p_t$. 
\paragraph{Complex noise.} In quantum state diffusion the noise enters through complex Wiener processes, and it is convenient to express \eqref{eq: stratonovich reversal} directly in complex form. Let $\xi_k = (W_{k1}+iW_{k2})/\sqrt2$, $k=1,\dots,K$, with independent real Wiener processes $W_{ka}$, so that $d\xi_k\,d\xi_k^* = dt$ and $d\xi_k^2 = 0$. A \emph{complex vector field} is a combination $Z = (V_1-iV_2)/\sqrt2$ of two real vector fields. It acts on functions by $Zf := (V_1f-iV_2f)/\sqrt2$, and $\bar Z := (V_1+iV_2)/\sqrt2$, so that $\bar Zf = \overline{Zf}$ for real $f$. Operations that are linear in the vector field are extended complex-linearly, in particular the divergence, $\mathrm{div}_pZ := (\mathrm{div}_pV_1 - i\,\mathrm{div}_pV_2)/\sqrt2 = \mathrm{div}\,Z + Z(\log p)$, and the directional derivative $(Z\cdot\partial)$ in coordinates. Since
$Z\,d\xi + \bar Z\,d\xi^*  = V_1\,dW_1 + V_2\,dW_2$, the Stratonovich SDE
\begin{align}
    dX = V_0\,dt + \sum_{k=1}^K\big(Z_k\circ d\xi_k + \bar Z_k\circ d\xi_k^*\big),
    \label{eq: complex stratonovich}
\end{align}
with a real drift $V_0$, is a Stratonovich SDE of the form considered above, with the $2K$ real noise fields $V_{k1}$, $V_{k2}$. It is then straightforward to translate everything so far into complex fields:

\begin{lemma}[Complex noise]\label{lem: complex reversal}
Consider the Stratonovich SDE with complex noise fields \eqref{eq: complex stratonovich}.
\begin{enumerate}
    \item Its generator is
    \begin{align}
        \mathcal{W}f = V_0f - \frac12\sum_k\Big[(\overline{\mathrm{div}\, Z_k})\,Z_k+(\mathrm{div}\, Z_k)\bar Z_k \Big]f +\frac{1}{2}\mathrm{div}(D\nabla f),
    \end{align}
    where
    \begin{align}
        D\nabla f = \sum_k\Big[\overline{(Z_kf)}\,Z_k + (Z_kf)\,\bar Z_k\Big].
        \label{eq: complex covariance}
    \end{align}
    \item The reverse  SDE has the same noise fields $Z_k$ and the drift
    \begin{align}
        \tilde V_0 = -V_0 + \sum_k\Big((\overline{\mathrm{div}_{p_t}Z_k})\,Z_k + (\mathrm{div}_{p_t}Z_k)\,\bar Z_k\Big) .
        \label{eq: complex divergence sum}
    \end{align}
\end{enumerate}
\end{lemma}

\section{Backward SDE}\label{sec: backward state vector}
\subsection{The geometry of pure states}\label{subsec: pure state geometry}
A pure state  can either be seen as a ray  of unit vectors in $\mathbb{C}^d$ or equivalently as a rank-1 projector $\psi = |\psi\rangle\langle\psi|$. The set of pure states is the complex projective space $\mathcal{P}\simeq\mathbb{CP}^{d-1}$, the quotient of the unit sphere by the global phases \cite{bengtsson2017geometry}. It is a compact manifold without boundary, of real dimension $2(d-1)$.  As reference measure we take the Fubini-Study volume, normalized to one. It is the unique probability measure on $\mathcal{P}$ invariant under all unitary transformations, and it is the distribution of $|v\rangle/\|v\|$ when the components of $|v\rangle\in\mathbb{C}^d$ are independent standard complex Gaussian variables \cite{bengtsson2017geometry}. 

A tangent vector at $\psi$ is the velocity $\delta\psi = \frac{d}{d\epsilon}\psi_\epsilon|_{\epsilon=0}$ of a curve on $\mathcal{P}$. Differentiating all properties of a rank-1 projector, namely $\psi_\epsilon^\dagger = \psi_\epsilon$, $\mathrm{tr}\,\psi_\epsilon = 1$ and $\psi_\epsilon^2 = \psi_\epsilon$, shows that tangent vectors are traceless Hermitian matrices $\delta\psi$ with  $\lbrace \psi,\delta\psi\rbrace= \delta\psi$ where $\lbrace, \rbrace$ denotes the anti-commutator.  The simplest curve is a linear flow generated by a matrix $M\in\mathbb{C}^{d\times d}$, followed by normalization. We denote the corresponding vector field by $X_M$
\begin{align}
    X_M(\psi) := \frac{d}{d\epsilon}\bigg|_{\epsilon=0}\frac{e^{\epsilon M}\psi\, e^{\epsilon M^\dagger}}{\mathrm{tr}\big(e^{\epsilon M}\psi\, e^{\epsilon M^\dagger}\big)} = M\psi+\psi M^\dagger - \mathrm{tr}\big(\psi(M+M^\dagger)\big)\psi .
    \label{eq: X_M}
\end{align}
The noise of Gisin's equation \eqref{eq: gisin projector stratonovich} is a field of this type. Denote the complex field $Z := B\psi$. Since $B = L-l$, it is easy to see that
\begin{align}
     Z = \big(X_L - i\,X_{iL}\big)/2.
    \label{eq: Z decomposition X}
\end{align}
This is not surprising: Gisin's equation preserves the norm of the state vector and Stratonovich calculus obeys the chain rule, so Stratonovich coefficients of Gisin's equation are tangent vector fields. The time reversal of a diffusion involves the divergence of its noise fields with respect to the reference measure. The following lemma characterizes this geometric quantity.
\begin{lemma}\label{lem: projective divergence}
For every $M\in\mathbb{C}^{d\times d}$, the divergence of $X_M$ with respect to the Fubini-Study measure is
\begin{align}
    \mathrm{div}\, X_M(\psi) = 2\,\mathrm{Re}\big(\mathrm{tr}M - d\,\langle\psi|M|\psi\rangle\big).
    \label{eq: projective divergence}
\end{align}
For the noise field of Gisin's equation, $\mathrm{div}\,Z = \mathrm{tr}L - d\, l$.
\end{lemma}

Let us finish with a simple remark that will be useful later on. For Hermitian $H$, \eqref{eq: X_M} becomes the  Schr\"odinger equation $\dot\psi = -i[H,\psi]$, which is the flow of $X_{-iH}$. The converse is also true locally. At a fixed $\psi$, every tangent vector $u$ is $X_{iH}(\psi)$ for some Hermitian $H$. Indeed $H_u := i[\psi,u]$ is Hermitian and traceless, and, using the properties of a tangent vector 
    \begin{align}
        X_{iH_u}(\psi) = i[H_u,\psi] = -\big[[\psi,u],\psi\big] = \psi u + u\psi - 2\,\psi u\psi = u .
        \label{eq: tangent rotation}
    \end{align}
Every tangent direction is therefore the velocity of a unitary rotation, which reflects the homogeneity of $\mathcal{P}$: the unitary group acts transitively on pure states, so its infinitesimal action already spans each tangent space. 
\subsection{Construction of the reverse equation}\label{subsec: construction}

From Section \ref{sec: qsd} we see that the Stratonovich form of Gisin's SDE for the projector \eqref{eq: gisin projector stratonovich} has the canonical form \eqref{eq: complex stratonovich} with 
\begin{align}
    V_0 = \Big(A-\frac12\bar l\Big)\psi + \psi\Big(A-\frac12\bar l\Big)^\dagger,\qquad Z_k = B_k\psi,\quad \bar Z_k = \psi B_k^\dagger.
    \label{eq: gisin fields}
\end{align}
By Lemma \ref{lem: complex reversal} and Lemma \ref{lem: projective divergence}, the reversal corrects the drift by $\sum_k\big(\overline{\Lambda_{k,t}}\,Z_k + \Lambda_{k,t}\,\bar Z_k\big)$, with 
\begin{align}
    \Lambda_{k,t}(\psi) := \mathrm{div}_{p_t}Z_k = \mathrm{tr}L_k - d\,l_k + \partial_{B_k}\log p_t(\psi),
    \label{eq: Lambda}
\end{align}
where we have introduced the notation $\partial_Bf := Zf$ for the derivative of a function along the noise fields. The reverse SDE for the projector is thus given by
\begin{align}
    d\psi = \tilde V_0\,d\tau + \sum_k\big(Z_k\circ d\tilde\xi_k + \bar Z_k\circ d\tilde\xi_k^*\big),\qquad \tilde V_0 = -V_0 + \sum_k\big(\overline{\Lambda_{k,t}}\,B_k\psi + \Lambda_{k,t}\,\psi B_k^\dagger\big) .
\end{align}
Inverting the chain rule $d\psi = |d\psi\rangle \langle \psi| + |\psi\rangle \langle d\psi|$, we obtain (up to a phase) the reverse Stratonovich SDE for the state vector
\begin{align}
    |d\psi\rangle = \Big[-\Big(A-\frac12\bar l\Big) + \sum_k\overline{\Lambda_{k,t}}\,B_k\Big]|\psi\rangle\,d\tau + \sum_kB_k|\psi\rangle\circ d\tilde\xi_k.
\end{align}
Since the noise fields of the backward and forward equation are identical, the It\^{o} correction is again $\bar l/2$, and the It\^{o} drift is $-A+\bar l+\sum_k\overline{\Lambda_{k,t}}B_k$. This proves the main result of this section.

\begin{theorem}[Backward quantum state diffusion]\label{thm: backward qsd}
In reverse time $\tau = T-t$, the time-reversed state obeys the It\^{o} SDE
\begin{align}
    |d\psi\rangle = \Big[-A + \bar{l} + \sum_k\overline{\Lambda_{k,t}}\, (L_k-l_k)\Big]|\psi\rangle\, d\tau + \sum_k(L_k-l_k)|\psi\rangle\, d\tilde{\xi}_k,
    \label{eq: backward qsd}
\end{align}
with $A$, $\bar l$, $\Lambda_{k,t}$ given in \eqref{eq: gisin coefficients}, \eqref{eq: l bar}, \eqref{eq: Lambda} and independent complex Wiener increments $d\tilde{\xi}_k$. Its Stratonovich drift is $-(A-\frac12\bar{l}) + \sum_k\overline{\Lambda_{k,t}}(L_k-l_k)$.
\end{theorem}

As a self-consistency check, we confirm that if the backward dynamics starts from $p_T$ then the ensemble average $\tilde\rho(\tau) := \mathbb{E}[\psi_\tau]$ follows the reverse Lindblad $-\mathcal{L}$ evolution.  The backward SDE of the projector reads in It\^o form
\begin{align}
    d\psi = \Big(-\mathcal{L}(\psi) + K_t(\psi)\Big)d\tau + \sum_k\big(B_k\psi\, d\tilde\xi_k + \psi B_k^\dagger d\tilde{\xi}_k^*\big),
    \label{eq: backward projector}
\end{align}
with 
\begin{align}
    K_t(\psi) = 2\Big(\sum_kB_k\psi B_k^\dagger + \bar{l}\,\psi\Big) + \sum_k\big(\overline{\Lambda_{k,t}}\,Z_k + \Lambda_{k,t}\,\bar Z_k\big).
\end{align}
At reverse time $\tau$, the law of $\psi_\tau$ is $p_t$ and we will show that $\int p_tK_t(\psi)dm=0$. By \eqref{eq: Lambda} and \eqref{eq: div p},  $p_t\Lambda_{k,t} = \mathrm{div}(p_tZ_k)$. Integrating by parts gives
\begin{align*}
\int\sum_k\big(p_t\overline{\Lambda_{k,t}}\,Z_k + p_t\Lambda_{k,t}\,\bar Z_k\big)\,dm = -\int p_t\sum_k\Big[(\bar Z_k\cdot\partial)Z_k + (Z_k\cdot\partial)\bar Z_k\Big]dm= -2\int p_t\Big(\sum_kB_k\psi B_k^\dagger + \bar l\,\psi\Big)\,dm ,
\end{align*}
where the last equality used the fact that $\sum_k[(\bar Z_k\cdot\partial)Z_k + (Z_k\cdot\partial)\bar Z_k]$ is twice the Stratonovich-It\^o correction, which can be read off from the forward equations \eqref{eq: projector sde} and \eqref{eq: gisin projector stratonovich} to be $\sum_kB_k\psi B_k^\dagger + \bar l\,\psi$. This shows that the expectation of $K_t$ vanishes. If we were to start from another ensemble, the law of $\psi_\tau$ would be some distribution $q$ different from $p_t$ and the expectation of $K_t$ no longer vanishes since $q\Lambda_{k,t}$ is not a total divergence.

\begin{corollary}[Reverse Lindbladian]\label{corr: reverse lindblad}
Starting from $p_T$, the density matrix of the backward ensemble obeys
\begin{align}
    \frac{d\tilde{\rho}}{d\tau} = -\mathcal{L}[\tilde\rho].
\end{align}
\end{corollary}

\section{Starting from a universal prior}\label{sec: prior}
In score-based generative models \cite{sohldickstein2015deep, ho2020denoising, song2021scorebased}, the forward process is designed to forget its initial distribution, the data distribution. After a long time $T$ its law $p_T$ is close to a stationary distribution $p_\infty$, a standard Gaussian in flat space, which does not depend on the data and is easy to sample. The reverse SDE, with the score of $p_t$, is then started from $p_\infty$ instead of $p_T$, and it produces approximate samples of $p_0$.  In this section we ask what the backward SDE \eqref{eq: backward qsd} produces when it is started from a different ensemble from $p_T$, in particular from a universal one. We first show that the unravelling of the depolarizing channel leads to a Brownian motion on $\mathbb{CP}^{d-1}$. We then bound the error of the distribution that arises from the backward SDE \eqref{eq: backward qsd} under various settings: exact scores and continuous time, learned scores and continuous time and finally learned scores and discrete time.

\subsection{The depolarizing channel}\label{subsec: depolarizing in d}
Let $\sigma_1,\dots,\sigma_{d^2-1}$ be an orthogonal basis of the traceless Hermitian $d\times d$ matrices normalized by $\mathrm{tr}(\sigma_j\sigma_k) = 2\delta_{jk}$, and let $H=0$ and $L_k = \sqrt\gamma\,\sigma_k$. We will show that the generator of the resulting diffusion process is proportional to the Laplace-Beltrami operator. From Lemma \ref{lem: complex reversal} we know that the generator of a Stratonovich SDE is given by
\begin{align}
    \mathcal{W}f = V_0f - \frac12\sum_k\Big[(\overline{\mathrm{div}\, Z_k})\,Z_k+(\mathrm{div}\, Z_k)\bar Z_k \Big]f +\frac{1}{2}\mathrm{div}(D\nabla f).
\end{align}
For Gisin's SDE for the projector \eqref{eq: gisin fields}, these are given by
\begin{align}
     V_0 =  \Big(A-\frac12\bar l\Big)\psi + \psi\Big(A-\frac12\bar l\Big)^\dagger,\qquad Z_k = B_k\psi,\quad \bar Z_k = \psi B_k^\dagger . 
\end{align}
Using the completeness relation of $\sigma_k$ matrices
\begin{align}
    \sum_k(\sigma_k)_{ij}(\sigma_k)_{ab} = 2\,\delta_{ib}\delta_{ja} - \frac2d\,\delta_{ij}\delta_{ab} ,
    \label{eq: completeness}
\end{align}
we can show the following identities for any matrix $X\in \mathbb{C}^{d\times d}$
\begin{align}
    \sum_k\sigma_kX\sigma_k = 2\,\mathrm{tr}(X)\,I - \frac2dX,\qquad \sum_k\mathrm{tr}(\sigma_kX)\,\sigma_k = 2X - \frac2d\,\mathrm{tr}(X)\,I,\qquad \sum_k\sigma_k^2 = \frac{2(d^2-1)}{d}\,I .
    \label{eq: completeness consequences}
\end{align}
From these identities and Lemma \ref{lem: projective divergence}, it is straightforward to show that both $V_0$ and $\sum_k(\overline{\mathrm{div}\, Z_k})\,Z_k$ are zero. Therefore
\begin{align*}
    \mathcal{W}f 
    = \frac12\,\mathrm{div}\big(D\nabla f\big).
\end{align*}
Since the Laplace-Beltrami operator is the divergence of the gradient, we just need to show that $D\nabla f$ is proportional to the gradient. The gradient is defined by $g(\mathrm{grad}\,f,u) = df(u)$ for every tangent vector $u$, where $g(u,v) := \mathrm{tr}(uv)$ is the metric induced by the Hilbert-Schmidt inner product of Hermitian matrices. The Riemannian volume of $g$ is unitarily invariant, so after normalization it is the Fubini-Study measure. Fix $\psi$ and a real function $f$, and let $G = \mathrm{grad}\,f$ at $\psi$. As a tangent vector, $G$ is a traceless Hermitian matrix with $\psi G + G\psi = G$. Extending $df(u) = \mathrm{tr}(Gu)$ complex-linearly to $Z_k$ gives $Z_kf = \mathrm{tr}(GB_k\psi) = \sqrt\gamma\,\mathrm{tr}(\sigma_k\psi G)$. By \eqref{eq: complex covariance},
\begin{align*}
    D\nabla f = \sum_k\Big[\overline{(Z_kf)}\,Z_k + (Z_kf)\,\bar Z_k\Big] = \gamma\sum_k\Big[\mathrm{tr}(\sigma_kG\psi)\,\big(\sigma_k-\langle\sigma_k\rangle\big)\psi + \mathrm{tr}(\sigma_k\psi G)\,\psi\big(\sigma_k-\langle\sigma_k\rangle\big)\Big].
\end{align*}
With help of the second identity in \eqref{eq: completeness consequences}, it can  be shown that
\begin{align}
    D\nabla f  =  2\gamma\,\mathrm{grad}\,f,\label{eq: D grad}
\end{align}
thus
\begin{align}
    \mathcal{W} = \gamma\,\Delta_g.
    \label{eq: depol brownian}
\end{align}
Quantum state diffusion for the depolarizing channel is thus a Brownian motion on $\mathbb{CP}^{d-1}$.  Let us point out a few consequences of this fact. 
First, by the divergence theorem, $\int\mathcal{W}f\,dm = \gamma\int\Delta_gf\,dm = 0$ for every $f$. Therefore the Fubini-Study measure is stationary, its density is $p_\infty=1$.
Second, by Green's identity, $\Delta_g$ is symmetric with respect to $dm$, $\mathcal{W}$ is thus self-adjoint and the process is reversible.
Third, the eigenvalues of $\mathcal{W}$ are $-\lambda_k$ with \cite{ikeda1978spectra, boucetta2007spectra}
    \begin{align}
        \lambda_k = 2\gamma\,k(k+d-1),\qquad k=0,1,2,\dots
        \label{eq: depol spectrum}
    \end{align}
The corresponding eigenspace $\mathcal{H}_k$ has dimension $(d-1)(d-1+2k)\big((d+k-2)!\big)^2/((d-1)!\,k!)^2$. The spectral gap is $\lambda_1 = 2\gamma d$; its eigenspace $\mathcal{H}_1$ has dimension $ d^2-1$ and consists of the traceless linear observables. Indeed, for traceless $O$, $\mathcal{W}\langle O\rangle = \langle\mathcal{L}^\dagger[O]\rangle = -2\gamma d\,\langle O\rangle$ by \eqref{eq: observable sde} and \eqref{eq: completeness consequences}, and the functions $\langle O\rangle$ with traceless Hermitian $O$ form a space of dimension $d^2-1$.  The modes $k\ge2$, polynomials of higher degree in $\psi$, describe the structure of the ensemble beyond $\rho$ and decay faster. Note that when we talk about the spectral gap in this context, we refer to the generator $\mathcal{W}$ that acts on functions on $\mathbb{CP}^{d-1}$. This should not be confused with the Lindblad operators $\mathcal{L}$ and $\mathcal{L}^\dagger$ that act on matrices. For the depolarizing channel 
\begin{align}
    \mathcal{L}[\rho] =  -2\gamma d\,\Big(\rho-\frac Id\Big),
    \label{eq: depol d}
\end{align}
so the two gaps coincide in this particular case but in general we have $\mathrm{gap}(\mathcal{W})\leq \mathrm{gap}(\mathcal{L})$.

To close this subsection, we write down the forward and backward SDE, which follows from the fact that both $V_0$ and $\sum_k(\overline{\mathrm{div}\, Z_k})\,Z_k$ are zero.
\begin{proposition}[Forward and backward SDE of the depolarizing channel]\label{prop: forward backward depol}
    
For the depolarizing channel, the forward Stratonovich SDE of $\psi$ is
\begin{align}
    d\psi =  \sum_k\big(B_k \psi\circ d\xi_k + \psi B_k^\dagger \circ d\xi_k^*\big),\label{eq: forward depol}
\end{align}
and the reverse SDE is 
\begin{align}
d\psi = 2\gamma \,\mathrm{grad}\, \log p_t\,d\tau+  \sum_k\big(B_k \psi\circ d\tilde \xi_k + \psi B_k^\dagger\circ d\tilde\xi_k^*\big).
\end{align}
\end{proposition}

\subsection{Reversal with exact score}\label{subsec: prior error}
In this subsection, we start the backward SDE from a density $q\neq p_T$ and we seek to bound the distance between $\hat p_0$ and $p_0$, where $\hat p_t$ denotes the density of the backward dynamics at the physical time $t\in[0,T]$. In Section \ref{sec: reversal} we have introduced the infinitesimal generators that govern instantaneous rate of change of functions and probability distributions; now in order to see how $\hat{p}_0$ depends on $q$, we use the exponential of these generators. We start with general $q$ before setting $q=p_\infty$.

Let $P_s$ be the transition operator of the forward process,
\begin{align}
    (P_sh)(x) := \mathbb{E}\big[h(X_{t+s})\,\big|\,X_t=x\big].
    \label{eq: transition operator}
\end{align}
For Gisin's equation the coefficients do not depend on time, so $P_s$ does not depend on $t$.
In the language of Section \ref{sec: reversal}, $P_s$ is the semigroup generated by $\mathcal{W}$, $P_s = e^{s\mathcal{W}}$. Its adjoint in $L^2(dm)$ transports densities, $p_{t+s} = P_s^*p_t$ with $P_s^* = e^{s\mathcal{W}^*}$, which is the integrated form of the Fokker-Planck equation $\partial_tp_t = \mathcal{W}^*p_t$. The two actions are related by
\begin{align}
    \int (P_sh)\,p\,dm = \int h\,(P_s^*p)\,dm .
    \label{eq: adjoint relation}
\end{align}
 Let $P_s(x,y)$ be the transition density i.e. the density of $X_{t+s}$ at $y$ given $X_t=x$, so that $(P_sh)(x) = \int P_s(x,y)\,h(y)\,dm(y)$. By Bayes' rule, the transition density of the backward process from reverse time $0$ to $\tau = T-t$ is
\begin{align}
    K_{T-t}(y,x) = \frac{p_t(x)\,P_{T-t}(x,y)}{p_T(y)} .
    \label{eq: backward kernel}
\end{align}
Integrating it against $q$ gives the following standard result
\begin{align}
    \hat p_t = p_t\cdot P_{T-t}\Big(\frac{q}{p_T}\Big).
    \label{eq: h transform}
\end{align}
 Two consequences are immediate. First, pure initial states are recovered exactly. Second, the  map $q\to  p_t\cdot P_{T-t}(q/p_T)$ is a classical channel i.e. a Markov operator on probability densities over $\mathcal{P}$: it is linear, it preserves positivity and mass. For any $f$-divergence $D$ (such as the total variation distance, the relative entropy and the $\chi^2$ divergence), the data-processing inequality says that
    \begin{align}
        D(\hat p_0\,\|\,p_0) \le  D(q\,\|\,p_T).
        \label{eq: data processing}
    \end{align}

Now consider the stationary density $q = p_\infty = 1$. Let us write $p_t = 1+\epsilon_t$, where $\int\epsilon_t\,dm = 0$. It can be shown that the deviation $\epsilon_t$ decays at the rate $\lambda_1$, uniformly in $\psi$ (see \eqref{eq: epsilon decay}). As a result, the data-processing inequality \eqref{eq: data processing} gives $\mathrm{TV}(\hat p_0,p_0)\le\mathrm{TV}(1,p_T) = O(e^{-\lambda_1T})$. However, the reversibility of the Brownian motion makes the actual decay rate larger than this.  We outline the main idea here and refer to the proof of Proposition \ref{prop: prior error} for full details. According to \eqref{eq: h transform}  
\begin{align}
    \frac{\hat p_0}{p_0} = P_T\Big(\frac{1}{p_T}\Big) = P_T( 1-\epsilon_T+\epsilon_T^2)+O(e^{-3\lambda_1T}).
\end{align}
Since $P$ is mass preserving, $P_T1 = 1$. From our discussion in the previous subsection, $P$ is symmetric (since $\mathcal{W}$ is self-adjoint) and as a result $P_T\epsilon_T = \epsilon_{2T}$, which is $O(e^{-2\lambda_1T})$. For  $\epsilon_T^2$, we split it into $\epsilon_T^2 = \|\epsilon_T\|^2 + h$ with its mean $\|\epsilon_T\|^2 = \int\epsilon_T^2\,dm$ and the fluctuation $h$. The mean is $O(e^{-2\lambda_1T})$ and is preserved by $P_T$ while the fluctuation, having no component in $\mathcal{H}_0$, is suppressed by  a factor of $e^{-\lambda_1T}$. The leading order of the relative error between $\hat p_0$ and $p_0$ is thus $-\epsilon_{2T}+ \|\epsilon_T\|^2$. To find the exact coefficient of $e^{-2\lambda_1T}$, we need the projection of $\epsilon_T$ into $\mathcal{H}_1$. As discussed above, the space $\mathcal{H}_1$ consists of the functions $\langle O\rangle$ with traceless Hermitian $O$. The projection $\Pi_1\epsilon_t$ is the unique element of $\mathcal{H}_1$ whose inner products with all elements of $\mathcal{H}_1$ agree with those of $\epsilon_t$. For traceless $O$, the latter are
\begin{align*}
    \int\epsilon_t\,\langle O\rangle\,dm = \int p_t\,\langle O\rangle\,dm - \int \langle O\rangle\,dm = \mathrm{tr}\big(O\rho(t)\big) - \frac{\mathrm{tr}\,O}{d} = \mathrm{tr}\Big(O\Big(\rho(t)-\frac Id\Big)\Big),
\end{align*}
where we used $\int\psi\,dm = I/d$ by unitary invariance. To find the function in $\mathcal{H}_1$ that matches these inner products, we need the second moment of the Fubini-Study measure \cite{harrow2013church},
\begin{align*}
    \int\psi\otimes\psi\;dm = \frac{I+\mathrm{SWAP}}{d(d+1)},
\end{align*}
where $\mathrm{SWAP}$ exchanges the two tensor factors. For traceless $X$ and $O$ it gives $\int \langle X\rangle\langle O\rangle\,dm = \mathrm{tr}\big((X\otimes O)(I+\mathrm{SWAP})\big)/d(d+1)  = \mathrm{tr}(XO)/d(d+1)$. 
\begin{align*}
    \Pi_1\epsilon_t(\psi) = d(d+1)\,\Big\langle\psi\Big|\rho(t)-\frac Id\Big|\psi\Big\rangle,\qquad \|\Pi_1\epsilon_t\|^2 = d(d+1)\,\mathrm{tr}\Big[\Big(\rho(t)-\frac Id\Big)^2\Big].
\end{align*}
Since  $\rho(t)-I/d=e^{-\lambda_1t}(\rho(0)-I/d)$, both terms are $\propto e^{-2\lambda_1T}$. Thus, to leading order, only the density matrix matters.  The structure of the ensemble that is invisible in $\rho$ lives in the modes $k\ge2$, which are forgotten faster.
\begin{proposition}[Forgetting the universal prior]\label{prop: prior error}
Consider the depolarizing channel in dimension $d$, with Lindblad gap $\lambda_1 = 2\gamma d$, and an initial ensemble $p_0$ with density matrix $\rho(0)$. Let $\hat p_0$ be the output of the backward SDE started from the universal prior. Then, as $T\to\infty$,
\begin{align}
    \frac{\hat p_0(\psi)}{p_0(\psi)} = 1 - d(d+1)\,e^{-2\lambda_1T}\Big(\langle\psi|\rho(0)|\psi\rangle - \mathrm{tr}\,\rho(0)^2\Big) + O(e^{-3\lambda_1T}),
    \label{eq: prior error}
\end{align}
where the remainder is bounded by a constant times $e^{-3\lambda_1T}$, with a constant that depends only on $d$ and $\gamma$: the bound is uniform in $\psi$ and in the initial ensemble.
\end{proposition}
Since the remainder is uniform in $\psi$, \eqref{eq: prior error} can be integrated against $p_0$. This gives the distances between $\hat p_0$ and $p_0$, which can be compared with the data-processing inequality \eqref{eq: data processing}. The proof is given in Appendix \ref{app: prior error}.
\begin{corollary}[Integrated error]\label{cor: prior error integrated}
Under the assumptions of Proposition \ref{prop: prior error}, as $T\to\infty$,
\begin{align}
    \mathrm{TV}(\hat p_0,p_0) &= \frac{d(d+1)}{2}\,e^{-2\lambda_1T}\,\mathbb{E}_{p_0}\Big|\langle\psi|\rho(0)|\psi\rangle - \mathrm{tr}\,\rho(0)^2\Big| + O(e^{-3\lambda_1T}),
    \label{eq: prior error tv}\\
    \mathrm{KL}(\hat p_0\,\|\,p_0) &= \frac{d^2(d+1)^2}{2}\,e^{-4\lambda_1T}\,\mathrm{Var}_{p_0}\big(\langle\psi|\rho(0)|\psi\rangle\big) + O(e^{-5\lambda_1T}),
    \label{eq: prior error kl}
\end{align}
whereas the right-hand sides of \eqref{eq: data processing} are $\mathrm{TV}(1,p_T) = O(e^{-\lambda_1T})$ and
\begin{align}
    \mathrm{KL}(1\,\|\,p_T) = \frac{d(d+1)}{2}\,e^{-2\lambda_1T}\Big(\mathrm{tr}\,\rho(0)^2-\frac1d\Big) + O(e^{-3\lambda_1T}).
    \label{eq: prior error dp}
\end{align}
\end{corollary}
Both distances thus decay at least twice as fast as the data-processing bound: at the rate $2\lambda_1$ in total variation and $4\lambda_1$ in relative entropy. 

\subsection{Learned score, continuous time}
In practice the score is not known exactly. It has to be learned from samples, for instance by score matching \cite{hyvarinen2005estimation, vincent2011connection, song2021scorebased}. Let $s_k(\psi,t)$ be complex-valued approximations of $\partial_{B_k}\log p_t$, and let $\hat p^{\,s}_t$ be the density of the solution of \eqref{eq: backward qsd} in which $\Lambda_{k,t}$ is replaced by $\mathrm{tr}L_k - d\,l_k + s_k(\psi,t)$, started at $\tau=0$ from a density $q$. The following bound transcribes to $\mathcal{P}$ the continuous-time bound of Song et al.\ for score-based generative models in flat space \cite[Theorem~1]{song2021maximum}.  Compared to their proof, we work with a compact manifold, so continuous scores give a bounded score error, and Novikov's condition and uniqueness hold automatically. In $\mathbb{R}^n$ they require additional assumptions \cite{song2021maximum} or an approximation argument \cite{chen2023sampling}. The proof is given in Appendix \ref{app: girsanov}.

\begin{proposition}[Error of a learned score]\label{prop: girsanov}
Assume that $p_t>0$, and that $\partial_{B_k}\log p_t(\psi)$ and $s_k(\psi,t)$ are continuous on $[0,T]\times\mathcal{P}$ and Lipschitz in $\psi$. Then
\begin{align}
    \mathrm{KL}\big(p_0\,\big\|\,\hat p^{\,s}_0\big) \le \mathrm{KL}\big(p_T\,\big\|\,q\big) + \int_0^T\mathbb{E}_{p_t}\Big[\sum_k\big|\partial_{B_k}\log p_t - s_k\big|^2\Big]\,dt .
    \label{eq: girsanov bound}
\end{align}
\end{proposition}

Regarding the training objective, the second term in \eqref{eq: girsanov bound} is a score-matching loss \cite{hyvarinen2005estimation}, which cannot be evaluated directly because $\partial_{B_k}\log p_t$ is unknown. As in flat space \cite{vincent2011connection, song2021scorebased}, it can be replaced by a denoising loss. The second term of \eqref{eq: girsanov bound} equals, up to a constant independent of $s$, the denoising loss
    \begin{align*}
        \int_0^T\mathbb{E}\Big[\sum_k\big|\partial_{B_k}\log P_t(\psi_0,\psi_t) - s_k(\psi_t,t)\big|^2\Big]dt .
    \end{align*}
For the depolarizing channel, unitary covariance and the invariance of $dm$ imply $P_t(U\psi'U^\dagger,U\psi U^\dagger) = P_t(\psi',\psi)$. Since any two states with the same fidelity from $\psi_0$ are related by a unitary fixing $\psi_0$, the stabilizer orbits are the level sets $\lbrace \psi: \mathrm{tr}(\psi \psi_0)=F\rbrace$. That is, $P_t(\psi_0,\psi)$ depends only on $F$. Since $F$ is the linear observable associated with the projector $\psi_0$,  the chain rule gives the explicit target $\partial_{B_k}\log P_t = \partial_F\log P_t\cdot\langle\psi|\psi_0(L_k-l_k)|\psi\rangle$.
    
We also note that if $p_0$ is a Dirac mass, as in Subsection \ref{subsec: depol d}, the score diverges as $t\to0$ and the assumptions fail at $t=0$. The bound then applies on $[\delta,T]$ and controls $\mathrm{KL}(p_{\delta}\,\|\,\hat p^{\,s}_{\delta})$.

\subsection{Learned score, discrete time} 

So far we have assumed that the backward SDE is solved in continuous time. In practice, a numerical solution lives on a discrete grid of times $T = t_0 > t_1 > \dots > t_N = \delta\ge0$, with steps $h_n = t_n-t_{n+1}$ and reverse times $\tau_n = T-t_n$.

From Proposition \ref{prop: forward backward depol} we know that the difference between the backward and forward SDE for the depolarizing channel is (up to a constant factor) just the gradient of the log density. In Subsection \ref{subsec: pure state geometry} we established that every tangent vector is the velocity of an associated Hamiltonian flow. Moreover, this Hamiltonian can be taken to be traceless and as such it can be expanded in terms of the high-dimensional Pauli matrices $\sigma_k$. The gradient is a tangent vector and by its very definition, the coefficients in this expansion are the derivatives along the direction of the corresponding rotation
\begin{align}
    \eta_{k,t}(\psi) := X_{i\sigma_k}\log p_t(\psi)= -\frac{2}{\sqrt\gamma}\,\mathrm{Im}\,\partial_{B_k}\log p_t(\psi).
    \label{eq: eta}
\end{align}
The following lemma is proved in Appendix \ref{app: depol structure}.
\begin{lemma}[Structure of backward depolarizing QSD]\label{lem: backward depol}
In Stratonovich form, the backward SDE for the projector reads
    \begin{align}
        d\psi = i\big[H_t(\psi),\psi\big]\,d\tau + \sum_k\big(B_k\psi\circ d\tilde\xi_k + \psi B_k^\dagger\circ d\tilde\xi_k^*\big),\qquad H_t(\psi) = \gamma\sum_k\eta_{k,t}(\psi)\,\sigma_k .
        \label{eq: depol hamiltonian}
    \end{align}
It is the forward equation \eqref{eq: forward depol}, plus a Hamiltonian that depends on the state through the score.
\end{lemma}

This decomposition into the forward dynamics and a score-dependent Hamiltonian has an important consequence on how we integrate the backward dynamics on a discrete time grid. On a reverse time interval, if we freeze the Hamiltonian at the beginning, then for the rest of that interval the frozen generator of \eqref{eq: depol hamiltonian} is split into a fixed rotation and a diffusion. This Lie-Trotter splitting is exact, because the two parts commute: the depolarizing channel is isotropic and rotations are symmetries of the forward dynamics.

In order to bound the error arising from this discrete time scheme using Girsanov's theorem, we need to bound the variation of the Hamiltonian coefficients within a time interval. A Noether-type argument can be used to show that these coefficients are martingales along the backward process: the derivative of the log-density along a symmetry generator is conserved in expectation along the reverse process. Then, by orthogonality of martingale increments, the variation of these coefficients within a time interval is bounded by the increment of the Fisher information along the noise fields. The metric here is an Einstein metric with positive Ricci curvature so the Li-Yau inequality \cite{li1986parabolic} can be used to bound the Fisher information.

\begin{lemma}[Score-dependent coefficients]\label{lem: depol structure}
For the depolarizing channel in dimension $d$ the following hold.
\begin{enumerate}

    \item  Along the backward process, $\eta_{k,T-\tau}(\psi_\tau)$, $0\le\tau\le T-\delta$, is a martingale for every $k$.
    \item The Fisher information along the noise,
    \begin{align}
        I(t) := \mathbb{E}_{p_t}\Big[\sum_{ij}D^{ij}\,\partial_i\log p_t\,\partial_j\log p_t\Big] = \gamma\,\mathbb{E}_{p_t}\Big[\sum_k\eta_{k,t}^2\Big],
    \end{align}
    is non-increasing in $t$, and $I(t)\le 2(d-1)/t$.
\end{enumerate}
\end{lemma}

Putting everything together, we have the following scheme to integrate the backward SDE for the depolarizing channel in discrete time and with approximate scores, along with a bound on the error of its solution. The scheme is the analogue of the exponential integrators used for score-based generative models \cite{zhang2023fast, chen2023sampling}, with the unitary rotation playing the role of the exactly integrated linear drift.
\begin{definition}[Scheme for discrete time, learned scores]\label{def: discrete time scheme}
Let $a_k(\psi,t)$ be bounded real approximations of $\eta_{k,t}$. Starting from $\hat\psi_0\sim q$, repeat for $n=0,\dots,N-1$ ($t_N=\delta\geq 0$):
\begin{enumerate}
    \item \emph{rotate}: $\hat\psi_n\to U_n\hat\psi_nU_n^\dagger$, with $U_n = e^{ih_nH_n}$ and $H_n = \gamma\sum_ka_k(\hat\psi_n,t_n)\,\sigma_k$;
    \item \emph{diffuse}: run the forward depolarizing QSD for a time $h_n$ from the rotated state, and call the result $\hat\psi_{n+1}$.
\end{enumerate}
\end{definition} 
 
\begin{proposition}[Time discretization]\label{prop: discretization}
For the depolarizing channel, let $\hat p_\delta$ be the law of $\hat\psi_N$ in Definition \ref{def: discrete time scheme}. Assume $\delta>0$, or $\delta=0$ with $p_0$ smooth and positive. Then
\begin{align}
    \mathrm{KL}\big(p_\delta\,\big\|\,\hat p_\delta\big) \le \mathrm{KL}\big(p_T\,\big\|\,q\big) + \gamma\sum_{n=0}^{N-1}h_n\,\mathbb{E}_{p_{t_n}}\Big[\sum_k\big(\eta_{k,t_n}-a_k(\cdot,t_n)\big)^2\Big] + \sum_{n=0}^{N-1}h_n\big[I(t_{n+1})-I(t_n)\big].
    \label{eq: discretization bound}
\end{align}
For the exact coefficients $a_k = \eta_{k,t}$, the second term vanishes and the last term can be multiplied by $\frac12$.
\end{proposition}

In this bound, the first sum is determined by the score error: refining the grid turns it into a continuous-time loss, but cannot remove it. On the other hand, the second sum vanishes as the steps shrink, and the choice of grid decides how fast, we refer to it as the discretization term. By Lemma \ref{lem: depol structure}, $I(t)$ grows at most like $2(d-1)/t$ as $t\to0$, a rate that is approached for a pure initial state, and for this rate the geometric grid is optimal.

\begin{corollary}[Geometric grid]\label{cor: step sizes}
If $h_n\le\kappa\,t_{n+1}$ for all $n$, then
\begin{align}
    \sum_{n=0}^{N-1}h_n\big[I(t_{n+1})-I(t_n)\big]\le 2\kappa\,(d-1)\Big(1+\log\frac T\delta\Big).
    \label{eq: step size sum}
\end{align}
The geometric grid $t_n = \max\big(T(1+\kappa)^{-n},\delta\big)$ satisfies the condition with $N = \lceil\log(T/\delta)/\log(1+\kappa)\rceil$ steps.
\end{corollary}

A few comments on the discretization term.

To make the discretization term smaller than $\varepsilon$ we can choose $\kappa = \varepsilon/[2(d-1)(1+\log(T/\delta))]$, which requires $N\simeq\log(T/\delta)/\kappa = 2(d-1)(1+\log(T/\delta))\log(T/\delta)/\varepsilon$ score evaluations. This is linear in the real dimension $2(d-1)$ of $\mathcal{P}$. It is the analogue of the nearly $d$-linear bounds for score-based generative models in flat space \cite{benton2024nearly}. There, too, a martingale property of the rescaled score along the reverse process is combined with a bound on its increments, which comes from a stochastic-localization identity for the posterior covariance. Here the martingale property is a consequence of unitary covariance, and the Li-Yau inequality bounds the increments.

The discretization term is proportional to the step size. As for Proposition \ref{prop: girsanov}, the bound controls the relative entropy between the laws of whole paths, and for this quantity the first order is natural: the martingale increments have variances proportional to the step. In flat space, for the score-based generative model built on the critically damped Langevin diffusion, the relative entropy between the path laws is bounded below by a multiple of the step size, even for Gaussian data and exact scores \cite[Theorem~7]{chen2023sampling}. The law at the final time is closer to the target: with the exact score, its relative entropy decreases faster than linearly in the examples of \S\ref{subsec: depolarizing} and \S\ref{subsec: depol d}, and approaches a quadratic dependence for small steps.

\section{Numerical illustration}\label{sec: examples}
We test the bounds of Section \ref{sec: prior} for the depolarizing channel with $\gamma = 1/4$, in dimensions $d = 2,3,4,6,8$, which include one, two and three qubits.

\subsection{Setting}\label{subsec: setting}
Fix a pure state $\psi_0$ and let $x := 2\,\mathrm{tr}(\psi_0\psi)-1\in[-1,1]$, an affine function of the fidelity with $\psi_0$.  The choice of $\psi_0$ does not matter: the depolarizing dynamics and the uniform prior are invariant under all unitaries, so every quantity reported below is the same for every $\psi_0$. We write $\nu$ for the law of $x$ under the Fubini-Study measure and $g_t := \partial_x\log p_t$ for the score in the variable $x$. All initial ensembles below are invariant under the unitaries that fix $\psi_0$. Then the score, and the output of the backward SDE, in continuous or discrete time and with exact or perturbed scores, are functions of $x$ alone. They can be computed by expanding in the eigenfunctions of the generator (Appendix \ref{app: numerics}). 

We use two initial ensembles. The first is the Dirac mass at $\psi_0$. Its Fisher information diverges as $t\to0$, and with the exact score the backward SDE recovers it exactly from any $q$ (\S\ref{subsec: prior error}), so it enters only the test of the time discretization in \S\ref{subsec: depol d}. The second is a smooth bimodal density \eqref{eq: smooth p0}. Its Fisher information at $t=0$ is finite, so the backward SDE can be run down to $t=0$.

\subsection{Exact score: forgetting the prior}\label{subsec: prior test}
We start the backward SDE with the exact score from the uniform prior $q=1$ and compare its output $\hat p_0$, given exactly by \eqref{eq: h transform}, with the smooth ensemble $p_0$.  Figure \ref{fig: prior} compares the leading terms of \eqref{eq: prior error tv} and \eqref{eq: prior error kl} with the exact distances and with the data-processing bounds. In every dimension the distances decay at the rates of Corollary \ref{cor: prior error integrated}: on $3\le\lambda_1T\le5$ the fitted rates are between $1.93\,\lambda_1$ and $2.00\,\lambda_1$ for the total variation distance, and between $3.86\,\lambda_1$ and $4.00\,\lambda_1$ for the relative entropy, while the data-processing bounds decay at the rates $\lambda_1$ and $2\lambda_1$.  

\begin{figure}[ht]
    \centering
    \includegraphics[width=\linewidth]{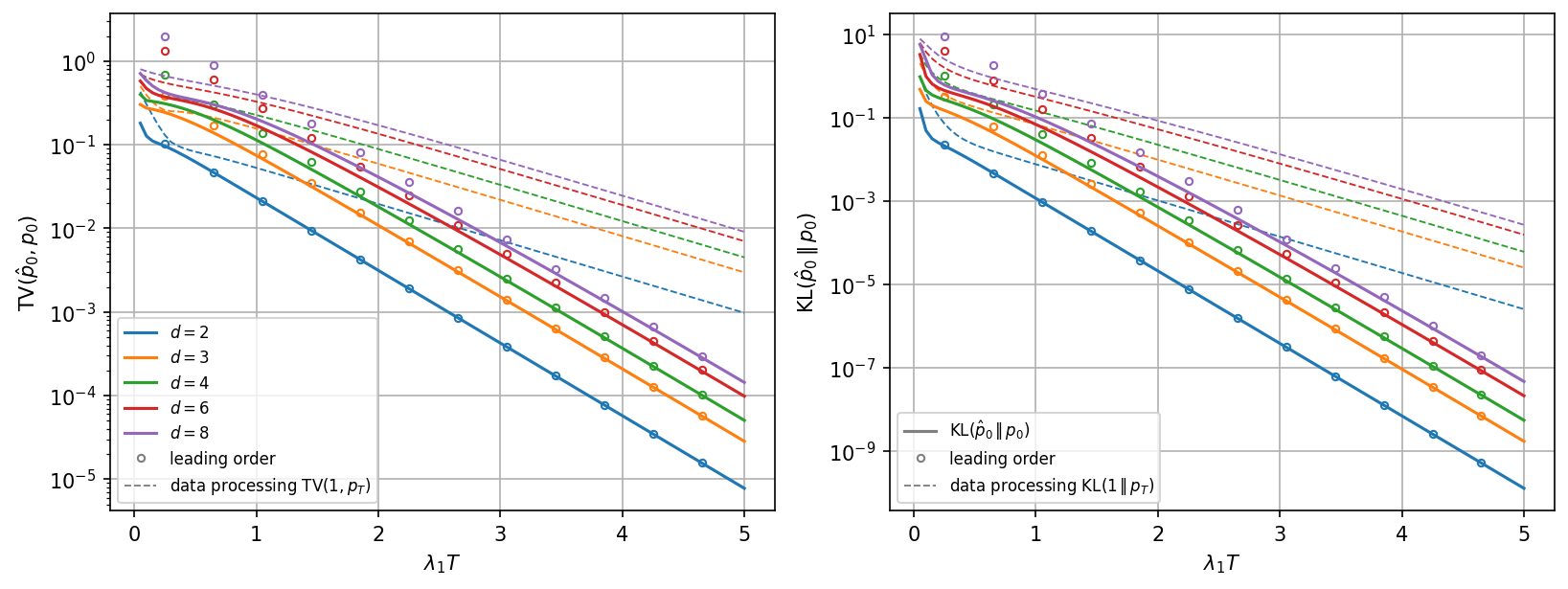}
    \caption{Forgetting the prior: depolarizing channel, $\gamma = 1/4$, smooth initial ensemble \eqref{eq: smooth p0}, backward SDE with the exact score started from the uniform prior $q=1$. Colours indicate the dimension $d$. \textbf{Left:} $\mathrm{TV}(\hat p_0,p_0)$ (solid lines), the leading term of \eqref{eq: prior error tv} (circles) and the data-processing bound $\mathrm{TV}(1,p_T)$ (dashed). \textbf{Right:} the same for the relative entropy, with $\mathrm{KL}(\hat p_0\,\|\,p_0)$, the leading term of \eqref{eq: prior error kl} and $\mathrm{KL}(1\,\|\,p_T)$. All as functions of $\lambda_1T$ with $\lambda_1 = 2\gamma d$.}
    \label{fig: prior}
\end{figure}

\subsection{Learned score, continuous time}\label{subsec: girsanov test}
We test Proposition \ref{prop: girsanov} with the smooth ensemble, $T=1$ and $q=1$. The exact scores are $\partial_{B_k}\log p_t = g_t\,\partial_{B_k}x$, and we replace them by learned scores $s_k = g^s_t\,\partial_{B_k}x$ with four kinds of error: a scale error, $g^s_t = 1.2\,g_t$; a constant bias, $g^s_t = g_t+2$; a state-dependent bias, $g^s_t = g_t+5x$; and a saturated score, $g^s_t = \mathrm{clip}(g_t,\pm10)$, as produced by a network with bounded output. The score-matching term of \eqref{eq: girsanov bound} is then $\int_0^T\mathbb{E}_{p_t}\big[2\gamma(1-x^2)(g_t-g^s_t)^2\big]dt$, and the law $\hat p^{\,s}_0$ of the output of the perturbed backward SDE is computed from a one-dimensional Fokker-Planck equation (Appendix \ref{app: numerics perturbed}). Figure \ref{fig: girsanov} shows that the bound holds in all cases; the ratio between the relative entropy and the bound is smallest for the state-dependent bias $g_t+5x$, between $19\%$ and $36\%$.

\begin{figure}[ht]
    \centering
    \includegraphics[width=0.95\linewidth]{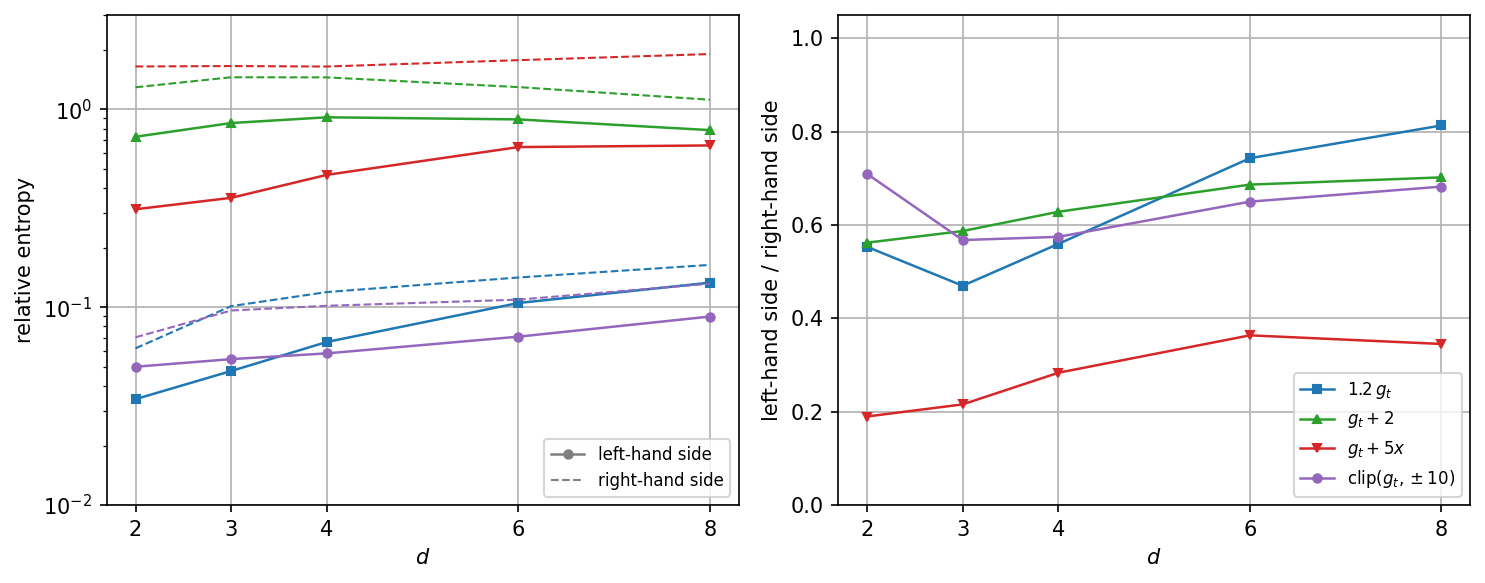}
    \caption{Learned score, continuous time: depolarizing channel, $\gamma = 1/4$, smooth initial ensemble \eqref{eq: smooth p0}, $T=1$, $q=1$, learned scores $s_k = g^s_t\,\partial_{B_k}x$. Colours and markers indicate the function $g^s_t$, as in the legend of the right panel. \textbf{Left:} relative entropy $\mathrm{KL}(p_0\,\|\,\hat p^{\,s}_0)$ of the final law (markers joined by solid lines) and right-hand side of \eqref{eq: girsanov bound} (dashed lines of the same colour). \textbf{Right:} ratio of the left-hand side to the right-hand side.}
    \label{fig: girsanov}
\end{figure}

\subsection{Time discretization: a smooth ensemble}\label{subsec: depolarizing}
We test Proposition \ref{prop: discretization} with the exact coefficients $a_k = \eta_{k,t}$ and $q = p_T$, so that only the discretization error remains. The bound \eqref{eq: discretization bound} then reduces to $\frac12\sum_nh_n\big[I(t_{n+1})-I(t_n)\big]$. For the smooth ensemble we take $T=1$, $\delta = 0$ and uniform steps $h$, for which the bound is $\frac12h\big[I(0)-I(1)\big]$. The scheme of Definition \ref{def: discrete time scheme} is run exactly on the law of $x$ (Appendix \ref{app: numerics scheme}). The results are shown in Figure~\ref{fig: discretization smooth}.

\begin{figure}[ht]
    \centering
    \includegraphics[width=0.95\linewidth]{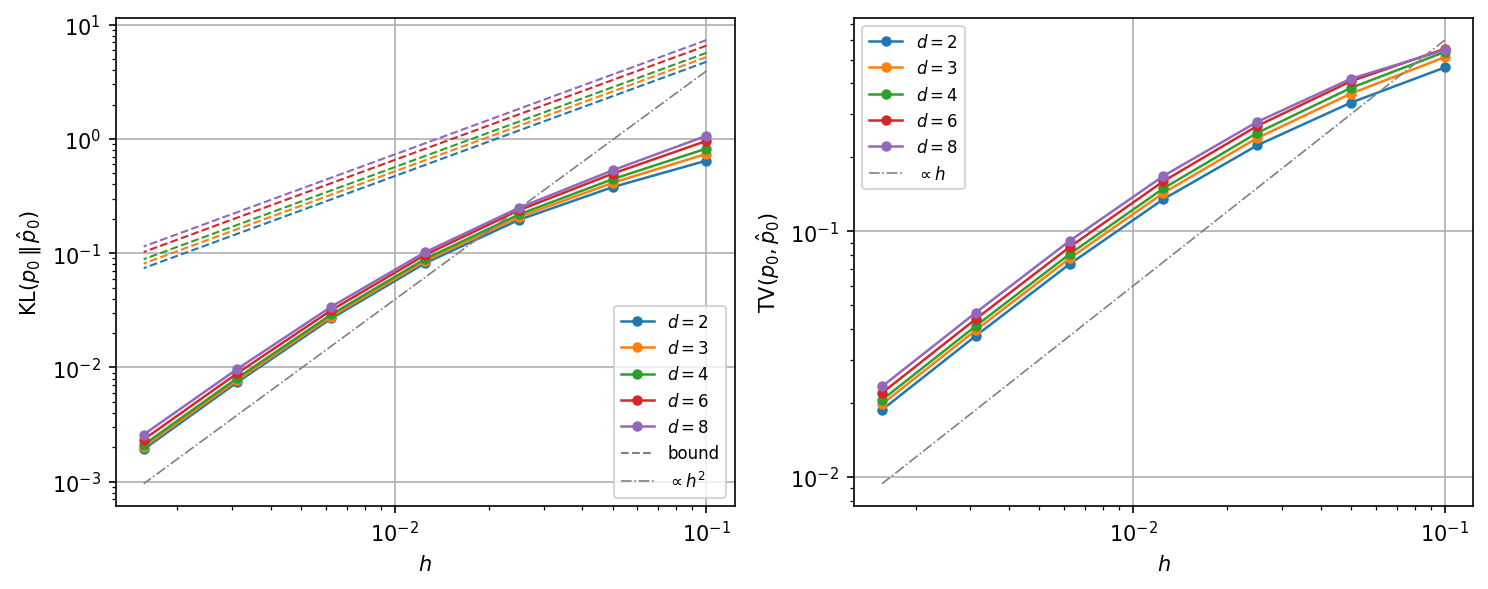}
    \caption{Time discretization for the smooth ensemble \eqref{eq: smooth p0}: depolarizing channel, $\gamma = 1/4$, exact score, started from $p_T$ with $T = 1$, uniform steps $h$ down to $\delta = 0$. \textbf{Left:} relative entropy $\mathrm{KL}(p_0\,\|\,\hat p_0)$ (dots) and bound \eqref{eq: discretization bound} with exact coefficients (dashed). \textbf{Right:} total variation distance $\mathrm{TV}(p_0,\hat p_0)$. Dash-dotted: slopes $2$ and $1$.}
    \label{fig: discretization smooth}
\end{figure}

While the bound is linear in the step size, the relative entropy of the final law is observed to decrease quadratically. Since the relative entropy is second-order in a small difference of densities, i.e. $\mathrm{KL}(p_0\,\|\,\hat p_0)\simeq\frac12\int(\hat p_0-p_0)^2/p_0\;dm$, its quadratic dependence on the step size can be explained if the relative error $\hat p_0/p_0-1$ is $O(h)$ uniformly. To partially support this hypothesis, we also plot the total variation distance in Figure \ref{fig: discretization smooth}. The results show a linear dependence on the step size.

With the perturbed scores considered above, the uniform prior $q=1$ and $h = 0.00625$, all three terms of \eqref{eq: discretization bound} contribute, and the bound holds in all cases.

\subsection{Time discretization: a pure initial state}\label{subsec: depol d}
For the Dirac mass at $\psi_0$ the Fisher information diverges as $t\to0$, and we use the geometric grids of Corollary \ref{cor: step sizes}. We again use the exact coefficients and start from $p_T$, with $T=2$, stop at $\delta = 0.1$, and take $\kappa = 0.4,0.2,0.1,0.05,0.025$, that is $N = 9$ to $122$ steps, the same in all dimensions. The bound is again \eqref{eq: discretization bound} with exact coefficients, and the corresponding a priori bound is $\kappa(d-1)(1+\log(T/\delta))$, half of the right-hand side of \eqref{eq: step size sum}. The results are shown in Figure~\ref{fig: discretization d}. We make three observations.

For all $d$, $I(t)$ is decreasing and $t\,I(t)$ stays below $2(d-1)$, as stated in Lemma \ref{lem: depol structure}. As $t$ decreases, $t\,I(t)$ increases toward $2(d-1)$. At $t=\delta$ it reaches $95\%$ of the Li-Yau value for $d=2$ and $81\%$ for $d=8$. 

For each $\kappa$, $\mathrm{KL}(p_\delta\,\|\,\hat p_\delta)/(d-1)$ is almost independent of $d$. A heuristic explanation is the flat-space analogue: for Brownian motion in $\mathbb{R}^n$ started from a Dirac mass, the density, the score and the exponential integrator all factorize over the coordinates, so the relative entropy is  proportional to $n$. The bound \eqref{eq: discretization bound} grows more slowly than $d-1$, because $I(t)$ falls further below the Li-Yau value as $d$ grows.

As for the smooth ensemble, $\mathrm{KL}(p_\delta\,\|\,\hat p_\delta)$ is observed to have a quadratic dependence on $\kappa$.

\begin{figure}[ht]
    \centering
    \includegraphics[width=0.95\linewidth]{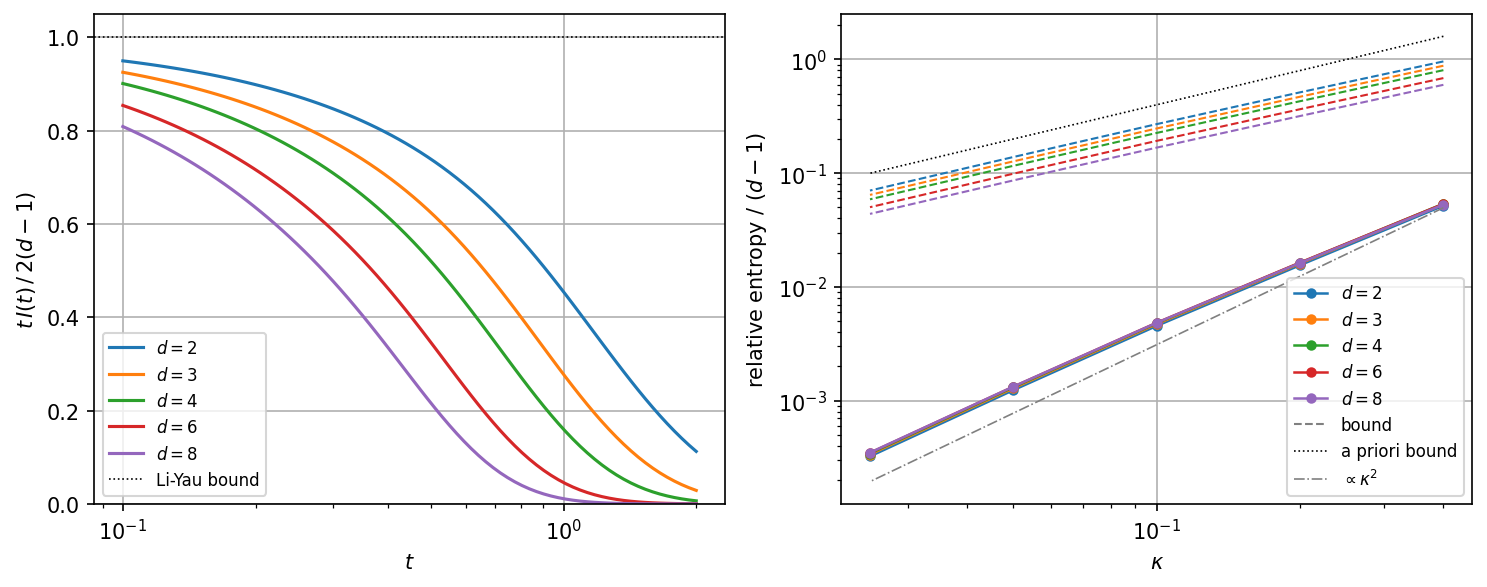}
    \caption{Time discretization for a pure initial state: depolarizing channel, $\gamma = 1/4$, Dirac mass $p_0$, exact score, started from $p_T$ with $T = 2$, $\delta = 0.1$. \textbf{Left:} $t\,I(t)$ in units of the Li-Yau value $2(d-1)$. \textbf{Right:} relative entropy $\mathrm{KL}(p_\delta\,\|\,\hat p_\delta)$ (dots) and bound \eqref{eq: discretization bound} with exact coefficients (dashed), divided by $d-1$, for the geometric grids with $\kappa = 0.4,0.2,0.1,0.05,0.025$. Dotted: a priori bound $\kappa(1+\log(T/\delta))$. Dash-dotted: slope $2$.}
    \label{fig: discretization d}
\end{figure}

\section{Discussion}\label{sec: discussion}
In this work we derive the time reversal of quantum state diffusion using an intrinsic, geometric formulation. The backward SDE for the state vector, Theorem \ref{thm: backward qsd}, has the same structure as Gisin's equation: the same noise operators, and a drift made of the forward drift with the opposite sign, a geometric term, and the score. At the level of the density matrix, the backward ensemble evolves with $-\mathcal{L}$ (Corollary \ref{corr: reverse lindblad}), only on the ensemble from which its score was computed. For the depolarizing channel the unravelling is a Brownian motion on $\mathbb{CP}^{d-1}$, so the Fubini-Study measure provides a universal prior in any dimension. Started from it, the backward SDE forgets the prior at twice the rate of the data-processing bound, $2\lambda_1$ in total variation and $4\lambda_1$ in relative entropy (Proposition \ref{prop: prior error}, Corollary \ref{cor: prior error integrated}). The errors due to a learned score and to a finite time step are controlled by Propositions \ref{prop: girsanov} and \ref{prop: discretization}; both bounds hold in all tested examples. In deriving these results we have combined tools and ideas from stochastic differential geometry, harmonic analysis on symmetric spaces, geometric analysis and group theory.

Several directions can be explored.

The score can be learned from simulated trajectories by score matching, as in score-based generative models on manifolds \cite{song2021scorebased, debortoli2022riemannian, huang2022riemanniandiffusionmodels}. In our setting, the divergence terms that such methods usually estimate stochastically are exactly computed in Lemma \ref{lem: projective divergence}. Moreover, for the depolarizing channel the denoising target is also available in closed form, as discussed below  Proposition \ref{prop: girsanov}. Given the success of score-based methods in classical machine learning, we hope that this work sets the stage for the use of classical models in the context of quantum state diffusion and also for the use of generative quantum models. 

Beyond the depolarizing channel, other channels with a smooth stationary density, whose density matrix relaxes to a state other than the maximally mixed one, would provide non-uniform priors. The reversibility, the explicit spectrum (or at least the spectral gap) and the decay rate at which such priors are forgotten remain to be determined.

Finally, Gisin's equation is one among several ways to unravel the master equation. Unravellings by quantum jumps \cite{dalibard1992wave, carmichael1993open} lead to piecewise-deterministic processes, the time reversal of which would require different tools.

As parting words, we hope to have convinced readers that the manifold of pure states have rich mathematical structures that make the diffusion processes on this manifold a worthwhile object of study. 

\section*{Acknowledgments}
The authors thank Haomu Yuan, Hitomi Mori and Bin Cheng for discussions in the early phase of this project. Claude Opus 5 and 5.5 were used to generate the numerical examples and help making the proofs more rigorous. This work is supported by the National Research Foundation, Singapore through the National Quantum Office, hosted in A*STAR, under its Centre for Quantum Technologies Funding Initiative and its Advanced Quantum Algorithms and Solutions Funding Initiative.
\bibliography{sample}
\appendix

\section{Proof of Lemma \ref{lem: projective divergence}}\label{app: projective divergence}
The divergence of a vector field measures the rate at which its flow expands or contracts volumes. If $\Phi_\epsilon$ is the flow of $X$ and $J_\epsilon(\psi)$ is the factor by which $\Phi_\epsilon$ multiplies the measure of a small neighbourhood of $\psi$, then $\mathrm{div}\,X(\psi) = \frac{d}{d\epsilon}J_\epsilon(\psi)|_{\epsilon=0}$. The flow of $X_M$ is particularly simple: by \eqref{eq: X_M}, it maps the ray $[z]$ of a vector $z\in\mathbb{C}^d$ to the ray $[gz]$, with $g = e^{\epsilon M}$. It is therefore enough to compute the volume factor $J_g$ of the action of an arbitrary invertible matrix $g$ on rays, and to differentiate at the end. We do so in three steps: we first realize the Fubini-Study measure as a volume in $\mathbb{C}^d$, we then compute how $g$ changes this volume, and we finally extract $J_g$.

\begin{proof}
Let $U\subset\mathcal{P}$ be a neighborhood of $\psi$. By definition, 
\begin{align*}
    J_g(\psi) = \lim_{U\to\psi}\frac{m(gU)}{m(U)}
\end{align*}
where $m$ is the Fubini-Study measure. This measure can be realized as the proportion of the cone of $U$ in the unit ball: $m(U) = \mathrm{vol}(C(U))/\mathrm{vol}(\text{unit ball})$, where $\mathrm{vol}$ is the Lebesgue volume of $\mathbb{C}^d\simeq\mathbb{R}^{2d}$ and $C(U) := \{z\in\mathbb{C}^d : 0<\|z\|\le1,\ [z]\in U\}$. Indeed, this defines a probability measure on $\mathcal{P}$, and it is unitarily invariant, because $C(W\cdot U) = W\,C(U)$ for a unitary $W$ and $W$ preserves the Lebesgue volume. Therefore
\begin{align*}
    \frac{m(gU)}{m(U)} = \frac{\mathrm{vol}(C(gU))}{\mathrm{vol}(C(U))} .
\end{align*}
To compare the volume of these two cones, we compare both of them with the intermediate volume of $g\,C(U)$, see Figure \ref{fig: cone volume} for an illustration.

\begin{figure}[ht]
    \centering
    \includegraphics[width=0.85\linewidth]{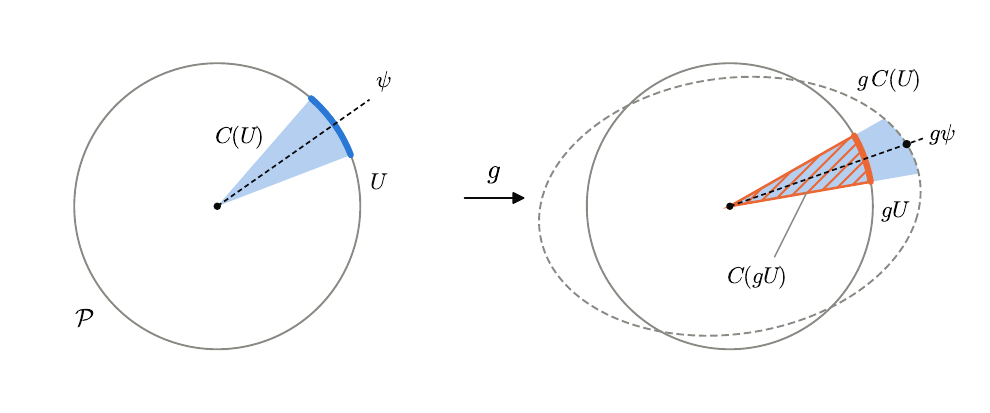}
    \caption{Illustration of the proof of Lemma \ref{lem: projective divergence}. \textbf{Left:} a set $U$ of rays around $\psi$ and its cone $C(U)$, whose area defines the measure $m(U)$. \textbf{Right:} the image $g\,C(U)$ (shaded) and the cone $C(gU)$ (hatched). For a very thin cone, the ratio between the volume of $g\,C(U)$ and that of $C(gU)$ is roughly the radius $\| g|\psi\rangle\|$ to the power of dimension, here $2d$.}
    \label{fig: cone volume}
\end{figure}

On one hand, viewed as a real linear map of $\mathbb{R}^{2d}$, $g$ has determinant $|\det g|^2$, since its realification has the eigenvalues $\lambda_j$ of $g$ together with their conjugates $\bar\lambda_j$. Hence 
\begin{align}
    \frac{\mathrm{vol}(g\,C(U))}{\mathrm{vol}(C(U))} = |\det g|^2 .
\end{align}
On the other hand, $g\,C(U) = \{w : [w]\in gU,\ \|g^{-1}w\|\le1\}$ is again a cone over $gU$, but it is no longer cut off at radius one. In polar coordinates $w = s\,u$ with $\|u\|=1$, where $d^{2d}w = s^{2d-1}ds\,d\sigma(u)$ and $\|g^{-1}w\| = s\|g^{-1}u\|$, its volume is
\begin{align*}
    \mathrm{vol}(g\,C(U)) = \int_{[u]\in gU}d\sigma(u)\int_0^{1/\|g^{-1}u\|}s^{2d-1}ds = \frac{1}{2d}\int_{[u]\in gU}\frac{d\sigma(u)}{\|g^{-1}u\|^{2d}} .
\end{align*}
Now let $U$ shrink to the ray of a unit vector $|\psi\rangle$. The integration variable $u$ then tends to $e^{i\chi}g|\psi\rangle/\|g|\psi\rangle\|$ for some phase $\chi$, so that $\|g^{-1}u\|\to1/\|g|\psi\rangle\|$, while the remaining integral $\frac{1}{2d}\int_{[u]\in gU}d\sigma(u)$ is $\mathrm{vol}(C(gU))$, by the same polar decomposition with the cut-off at radius one. Thus 
\begin{align}
    \lim_{U\to\psi} \frac{\mathrm{vol}(g\,C(U))}{\mathrm{vol}(C(gU))} = \|g|\psi\rangle\|^{2d} .
\end{align}
Combining these results, we find
\begin{align}
    J_g(\psi) = \frac{|\det g|^2}{\|g|\psi\rangle\|^{2d}} .
\end{align}
We now differentiate along $g = e^{\epsilon M}$. By Jacobi's formula, $|\det g|^2 = e^{2\epsilon\,\mathrm{Re}\,\mathrm{tr}M} = 1+2\epsilon\,\mathrm{Re}\,\mathrm{tr}M+O(\epsilon^2)$. Moreover, $\|g|\psi\rangle\|^2 = 1+2\epsilon\,\mathrm{Re}\langle M\rangle+O(\epsilon^2)$, so that $\|g|\psi\rangle\|^{2d} = 1+2d\,\epsilon\,\mathrm{Re}\langle M\rangle+O(\epsilon^2)$. Differentiating $J_g$ at $\epsilon=0$ gives \eqref{eq: projective divergence}. Two special cases serve as sanity checks: $\mathrm{div}\,X_{iH} = 0$ for Hermitian $H$, since unitaries preserve the measure, and $\mathrm{div}\,X_I = 0$, since $M=I$ does not move rays.
\end{proof}

\section{Proof of Proposition \ref{prop: prior error} and Corollary \ref{cor: prior error integrated}}\label{app: prior error}
The proof expands the outline of \S\ref{subsec: prior error}. The main task is to control the remainders uniformly in $\psi$ and in the initial ensemble, which can be done with the spectral expansion of the transition operators.

\begin{proof}
By \eqref{eq: h transform} with $t=0$ and $q=1$, the quantity to estimate is
\begin{align*}
    \frac{\hat p_0}{p_0} = P_T\Big(\frac{1}{p_T}\Big).
\end{align*}
For large $T$ the density $p_T = 1+\epsilon_T$ is close to one, and we can expand
\begin{align*}
    \frac{1}{p_T} = \frac{1}{1+\epsilon_T} = 1-\epsilon_T+\epsilon_T^2-\epsilon_T^3+\frac{\epsilon_T^4}{p_T} .
\end{align*}
We need to know how fast $\epsilon_T$ decays and how $P_T$ acts on functions such as $\epsilon_T$ and $\epsilon_T^2$.  In the following, all constants depend only on $d$ and $\gamma$, and $O(\cdot)$ is uniform in $\psi$ and in $p_0$. The estimates rely on three properties of the transition operators $P_s$.
\begin{enumerate}
    \item[(a)] They are Markov operators: $P_s1 = 1$ and $\sup|P_sh|\le\sup|h|$.
    \item[(b)] They are symmetric, $P_s(x,y) = P_s(y,x)$, because the process is reversible (\S\ref{subsec: depolarizing in d}). Hence they transport densities, $p_{t+s} = P_sp_t$; in particular $P_Tp_T = p_{2T}$.
    \item[(c)] They have the spectral expansion $P_s(x,y) = \sum_{k\ge0}e^{-\lambda_ks}\Pi_k(x,y)$ \cite{Bakry_2014}, where $\Pi_k(x,y)$ is the kernel of the orthogonal projection $\Pi_k$ onto the eigenspace $\mathcal{H}_k$. The dimension $N_k$ of $\mathcal{H}_k$ is given in \S\ref{subsec: depolarizing in d}. The eigenspace $\mathcal{H}_0$ consists of the constants, so $\Pi_0h = \int h\,dm$, which vanishes for functions of zero mean.
\end{enumerate}
To use the expansion (c) uniformly, we need a bound on the kernels $\Pi_k(x,y)$. The addition theorem for compact homogeneous spaces \cite{gine1975addition, koornwinder1973addition} gives $\Pi_k(x,x) = N_k$ for every $x$. The Cauchy-Schwarz inequality then gives $|\Pi_k(x,y)|\le N_k$. The expansion thus converges uniformly for $s>0$, because $N_k$ grows polynomially and $\lambda_k$ quadratically in $k$. 

The deviation $\epsilon_t$ has zero mean, so its expansion involves only $k\geq 1$ modes, i.e. $\epsilon_t = \sum_{k\ge1}e^{-\lambda_kt}\,\Pi_kp_0$. Since $\sup|\Pi_kp_0|\le N_k$ for every initial ensemble $p_0$,  $\epsilon_t$ decays at the rate $\lambda_1$ while contribution beyond the first eigenspace decays at the rate $\lambda_2$
\begin{align}
\sup|\epsilon_t|\le S_1(t),\qquad \sup\Big|\sum_{k\ge2}\Pi_k\epsilon_t\Big|\le S_2(t),\label{eq: epsilon decay}
\end{align}
where $S_j(t) := \sum_{k\ge j}N_ke^{-\lambda_kt}$.  More generally, $P_s$ damps every bounded function  $h$ of zero mean by a factor of order $e^{-\lambda_1s}$
\begin{align}
 \sup|P_sh|\le S_1(s)\,\sup|h|.
    \label{eq: sup decay}
\end{align}

With these estimates we can apply $P_T$ to the expansion of $1/p_T$. By \eqref{eq: epsilon decay}, the last two terms are $O(e^{-3\lambda_1T})$, and they stay so under $P_T$ by (a). We apply $P_T$ to the first three terms in turn. First, $P_T1 = 1$. Second, $P_T\epsilon_T = P_Tp_T-P_T1 = \epsilon_{2T}$ by (a) and (b). Third, we split $\epsilon_T^2 = \|\epsilon_T\|^2+h$ into its mean $\|\epsilon_T\|^2$, which is the $\chi^2$ divergence between $p_T$ and $p_\infty$ and is left unchanged by $P_T$, and the fluctuation $h$. The fluctuation has zero mean and $\sup|h|\le2\sup|\epsilon_T|^2 = O(e^{-2\lambda_1T})$, so $P_T$ damps it by a further factor $e^{-\lambda_1T}$ by \eqref{eq: sup decay}. Altogether,
\begin{align}
    \frac{\hat p_0}{p_0} = P_T\Big(\frac{1}{p_T}\Big) = 1-\epsilon_{2T}+\|\epsilon_T\|^2+O(e^{-3\lambda_1T}).
\end{align}
It remains to identify the terms of order $e^{-2\lambda_1T}$. By \eqref{eq: epsilon decay}, the components $\Pi_k\epsilon_t$ with $k\ge2$ are together $O(e^{-\lambda_2t})$. They therefore contribute $O(e^{-2\lambda_2T})$ both to $\epsilon_{2T}$ and to $\|\epsilon_T\|^2 = \sum_k\|\Pi_k\epsilon_T\|^2$, which is smaller than $e^{-3\lambda_1T}$ since $\lambda_2\ge2\lambda_1$. Only the first eigenspace survives, and its contribution was computed in \S\ref{subsec: prior error}:
\begin{align*}
    -\Pi_1\epsilon_{2T}+\|\Pi_1\epsilon_T\|^2 = d(d+1)\Big(\mathrm{tr}\Big[\Big(\rho(T)-\frac Id\Big)^2\Big] - \Big\langle\psi\Big|\rho(2T)-\frac Id\Big|\psi\Big\rangle\Big).
\end{align*}
Since $\rho(t)-I/d = e^{-\lambda_1t}(\rho(0)-I/d)$ by \eqref{eq: depol d}, the right-hand side is the correction in \eqref{eq: prior error}. This proves Proposition \ref{prop: prior error}.

For Corollary \ref{cor: prior error integrated}, we integrate the expansion against $p_0$, which is possible because its remainder is uniform. Let $R := \hat p_0/p_0-1$ and let $r$ be the correction in \eqref{eq: prior error}, so that $R = r+O(e^{-3\lambda_1T})$ and $r = O(e^{-2\lambda_1T})$, uniformly in $\psi$. Since $\mathrm{TV}(\hat p_0,p_0) = \frac12\mathbb{E}_{p_0}|R|$, integrating the uniform expansion directly gives \eqref{eq: prior error tv}. The relative entropy requires more care, because its expansion in $R$ starts at first order and this term cancels: both $\hat p_0$ and $p_0$ have unit mass, so $\mathbb{E}_{p_0}R = 0$. Likewise $\mathbb{E}_{p_0}r = 0$, because $\mathbb{E}_{p_0}\langle\psi|\rho(0)|\psi\rangle = \mathrm{tr}\,\rho(0)^2$, and therefore $\mathbb{E}_{p_0}[r^2]$ is $d^2(d+1)^2e^{-4\lambda_1T}$ times the variance in \eqref{eq: prior error kl}. With $(1+R)\log(1+R) = R+\frac12R^2+O(R^3)$, we obtain
\begin{align*}
    \mathrm{KL}(\hat p_0\,\|\,p_0) = \mathbb{E}_{p_0}\big[(1+R)\log(1+R)\big] = \frac12\mathbb{E}_{p_0}[R^2]+O(e^{-6\lambda_1T}) = \frac12\mathbb{E}_{p_0}[r^2]+O(e^{-5\lambda_1T}),
\end{align*}
which is \eqref{eq: prior error kl}.

Finally, to compare with the data-processing inequality, we evaluate the right-hand sides of \eqref{eq: data processing}. By \eqref{eq: epsilon decay}, $\mathrm{TV}(1,p_T) = \frac12\int|\epsilon_T|\,dm= O(e^{-\lambda_1T})$. For the relative entropy, the first-order term cancels again because $\epsilon_T$ has zero mean, and $\mathrm{KL}(1\,\|\,p_T) = -\int\log(1+\epsilon_T)\,dm = \frac12\|\epsilon_T\|^2+O(e^{-3\lambda_1T})$. As before, $\|\epsilon_T\|^2 = \|\Pi_1\epsilon_T\|^2+O(e^{-2\lambda_2T})$, this gives \eqref{eq: prior error dp} and completes the proof of Corollary \ref{cor: prior error integrated}.
\end{proof}

\section{Proof of Proposition \ref{prop: girsanov}}\label{app: girsanov}
The proof follows the flat-space argument of Song et al.\ \cite{song2021maximum}. On the compact manifold $\mathcal{P}$ all coefficients are bounded, so no additional integrability assumptions are needed.

\begin{proof}
We have to bound the relative entropy between $p_0$ and $\hat p^{\,s}_0$, the laws at the final time of two backward processes: the exact one, started from $p_T$, and the approximate one, started from $q$.  Let $Q$ and $Q^s$ be the laws of the paths $(\psi_\tau)_{\tau\in[0,T]}$ of the two processes. The map that sends a path to its endpoint has the law $p_0$ under $Q$ and $\hat p^{\,s}_0$ under $Q^s$, so the data-processing inequality \cite{cover2006elements, liese2006divergences} gives
\begin{align*}
    \mathrm{KL}\big(p_0\,\|\,\hat p^{\,s}_0\big)\le\mathrm{KL}\big(Q\,\|\,Q^s\big).
\end{align*}
The two path laws differ through their starting points and through their dynamics, and the chain rule of relative entropy \cite{cover2006elements} separates the two. Let $Q_x$ and $Q^s_x$ be the path laws of the exact and the approximate process started from a point $x$, so that $Q = \int p_T(x)Q_x\,dm(x)$ and $Q^s = \int q(x)Q^s_x\,dm(x)$. Then
\begin{align*}
    \mathrm{KL}\big(Q\,\|\,Q^s\big) = \mathrm{KL}\big(p_T\,\|\,q\big)+\int p_T(x)\,\mathrm{KL}\big(Q_x\,\|\,Q^s_x\big)\,dm(x).
\end{align*}
The first term is  the first term of \eqref{eq: girsanov bound}, so we just need to show that, for every starting point $x$,
\begin{align}
    \mathrm{KL}\big(Q_x\,\|\,Q^s_x\big)\le\mathbb{E}_{Q_x}\Big[\sum_k\int_0^T|\delta_k|^2\,d\tau\Big],\qquad \delta_k(\psi,t) := \partial_{B_k}\log p_t(\psi)-s_k(\psi,t),
    \label{eq: girsanov path bound}
\end{align}
with $\delta_k$ evaluated at $(\psi_\tau,T-\tau)$. Indeed, integrating \eqref{eq: girsanov path bound} against $p_T$ gives the expectation of the same quantity under $Q$, and since under $Q$ the state $\psi_\tau$ has the law $p_{T-\tau}$, this expectation is $\int_0^T\mathbb{E}_{p_t}\big[\sum_k|\delta_k|^2\big]dt$, the second term of \eqref{eq: girsanov bound}.

For a fixed starting point, the two processes have the same noise and different drifts, which is the setting of Girsanov's theorem. Before applying it, we need to make sure that both processes are well defined. All coefficients of the backward SDE depend on the state only through the projector, so we work with the backward SDE \eqref{eq: backward projector} for $\psi$. It is an It\^{o} SDE in the real vector space of Hermitian matrices whose solutions stay on $\mathcal{P}$, and its coefficients can be extended off $\mathcal{P}$ as Lipschitz functions without changing the solutions. By compactness and the assumptions, the score error $\delta_k$ is bounded, both SDEs have unique strong solutions, and pathwise uniqueness gives uniqueness in law \cite{oksendal2003stochastic, karatzas1991brownian}. 

Since only the $\Lambda_{k,t}$ depend on the score, the two drifts differ by $K_t-K^s_t = \sum_k\big(\bar\delta_k\,B_k\psi+\delta_k\,\psi B_k^\dagger\big)$.
This drift difference lies along the noise fields, so it can be absorbed into the noise. Let $\psi_\tau$ solve the exact SDE started from $x$ under a probability measure $\mathbb{P}$, driven by complex Wiener processes $\tilde\xi_k$. Then
\begin{align*}
    d\psi = \big(-\mathcal{L}[\psi]+K^s_t\big)d\tau + \sum_k\big(B_k\psi\,d\xi^s_k+\psi B_k^\dagger\,d\xi^{s*}_k\big),\qquad d\xi^s_k := d\tilde\xi_k+\bar\delta_k\,d\tau .
\end{align*}
In other words, the same path solves the approximate SDE, driven by the shifted processes $\xi^s_k$.  In terms of the real  components of $\tilde\xi_k = (W_{k1}+iW_{k2})/\sqrt2$, the shift is $W_{ka}\to W_{ka}+\int\theta_{ka}\,d\tau$ with
\begin{align*}
    \theta_{k1} = \sqrt2\,\mathrm{Re}\,\delta_k,\qquad \theta_{k2} = -\sqrt2\,\mathrm{Im}\,\delta_k,\qquad \sum_a\theta_{ka}^2 = 2|\delta_k|^2 .
\end{align*}
The process $\theta$ is adapted and bounded, so Novikov's condition holds, and Girsanov's theorem \cite{oksendal2003stochastic, karatzas1991brownian} provides a measure $\mathbb{P}^s$ with
\begin{align*}
    \log\frac{d\mathbb{P}}{d\mathbb{P}^s} = \sum_{k,a}\int_0^T\theta_{ka}\,dW_{ka} + \frac12\sum_{k,a}\int_0^T\theta_{ka}^2\,d\tau ,
\end{align*}
under which the $\xi^s_k$ are independent complex Wiener processes. By uniqueness in law, the law of $\psi$ under $\mathbb{P}^s$ is therefore $Q^s_x$. Under $\mathbb{P}$ and $\mathbb{P}^s$, the same map from the underlying probability space to paths thus has the laws $Q_x$ and $Q^s_x$. The data-processing inequality again gives us
\begin{align*}
    \mathrm{KL}\big(Q_x\,\|\,Q^s_x\big)\le\mathrm{KL}\big(\mathbb{P}\,\|\,\mathbb{P}^s\big) = \frac12\,\mathbb{E}_{\mathbb{P}}\Big[\sum_{k,a}\int_0^T\theta_{ka}^2\,d\tau\Big] = \mathbb{E}_{\mathbb{P}}\Big[\sum_k\int_0^T|\delta_k|^2\,d\tau\Big],
\end{align*}
which is \eqref{eq: girsanov path bound}, since the law of $\psi$ under $\mathbb{P}$ is $Q_x$. Note that only the form of the drift difference was used in this step. The same bound therefore holds for any drift difference $\sum_{k,a}\theta_{ka}V_{ka}(\psi_\tau)$ along the noise fields, with bounded adapted coefficients $\theta$, a fact that will be useful in Appendix \ref{app: discretization}.
\end{proof}

\section{Proof of Lemmas \ref{lem: backward depol} and \ref{lem: depol structure}}\label{app: depol structure}

\begin{proof}[Proof of Lemma \ref{lem: backward depol}] As explained in the main text, the idea is to expand the gradient in the frame of tangent vectors $X_{i\sigma_k}$. Let $G := \mathrm{grad}\,f$ be the gradient of a function $f$ at $\psi$. By \eqref{eq: tangent rotation}, $G$ can be regarded as the velocity of a rotation, $G = i[H_G,\psi]$, generated by the traceless Hamiltonian $H_G := i[\psi,G]$. Using the second identity of \eqref{eq: completeness consequences}, we can write $H_G$  as linear combination of $\sigma_k$ matrices,
\begin{align*}
    H_G = \frac{1}{2}\sum_k\mathrm{tr}(\sigma_kH_G)\,\sigma_k .
\end{align*}
Moreover, by definition of the gradient, the derivative of $f$ along $X_{i\sigma_k}$ is the inner product between $G$ and $X_{i\sigma_k}$,
\begin{align*}
    X_{i\sigma_k}f = \mathrm{tr}\big(G\,i[\sigma_k,\psi]\big) = \mathrm{tr}\big(\sigma_k\,i[\psi,G]\big) = \mathrm{tr}(\sigma_kH_G).
\end{align*}
Therefore, we have shown that, for every function $f$
\begin{align}
\mathrm{grad}\,f= i[H_G,\psi] = \frac{i}{2}\Big[\sum_k\mathrm{tr}(\sigma_kH_G)\,\sigma_k,\,\psi\Big]=\frac{1}{2}\sum_k\big(X_{i\sigma_k}f\big)\,X_{i\sigma_k} .\label{eq: tight frame}
\end{align}
Taking $f = \log p_t$ and we recover Lemma \ref{lem: backward depol}.
\end{proof}

\begin{proof}[Proof of Lemma \ref{lem: depol structure}]
We start with the martingale property. To alleviate the notation we fix $k$ and write $X := X_{i\sigma_k}$ and $f_t := \log p_t$, so that $\eta_{k,t} = Xf_t$. At reverse time $\tau$ the backward process has the generator $\tilde{\mathcal{W}}_t$, with $t = T-\tau$. By It\^o's formula applied to $u(\tau,\psi) := (Xf_{T-\tau})(\psi)$, the process $\eta_{k,T-\tau}(\psi_\tau)$ is the sum of its initial value, the time integral of the drift $(\partial_\tau+\tilde{\mathcal{W}}_t)u = (\tilde{\mathcal{W}}_t-\partial_t)(Xf_t)$, and a stochastic integral. We will show that the drift vanishes
\begin{align}
    \big(\tilde{\mathcal{W}}_t-\partial_t\big)(Xf_t) = 0 .
    \label{eq: no drift}
\end{align}
To evaluate the left-hand side, we need to know how the backward generator acts and how $f_t$ changes in time. Both can be best expressed with the carr\'e du champ operator $\Gamma(h,h') := \sum_{ij}D^{ij}\partial_ih\,\partial_jh'$ \cite{Bakry_2014}, with $\Gamma(h) := \Gamma(h,h)$. By \eqref{eq: generator}, it measures the failure of $\mathcal{W}$ to obey the Leibniz rule, $\mathcal{W}(hh') = h\,\mathcal{W}h'+h'\,\mathcal{W}h+\Gamma(h,h')$.

For the backward generator, \eqref{eq: nelson} and the product rule $p^{-1}\mathrm{div}(pV) = \mathrm{div}\,V+V\log p$ give $\tilde{\mathcal{W}}_th = -\mathcal{W}h+\mathrm{div}(D\nabla h)+\Gamma(f_t,h)$. Since $\mathcal{W} = \frac12\mathrm{div}(D\nabla\,\cdot\,)$,
\begin{align}
    \tilde{\mathcal{W}}_th = \mathcal{W}h+\Gamma(f_t,h) .
    \label{eq: step C}
\end{align}
For the time dependence, the process is reversible, so $\mathcal{W}^* = \mathcal{W}$ and the Fokker-Planck equation reads $\partial_tp_t = \mathcal{W}p_t$. A diffusion generator obeys the chain rule $\mathcal{W}\phi(h) = \phi'(h)\,\mathcal{W}h+\frac12\phi''(h)\,\Gamma(h)$ \cite{Bakry_2014}. With $\phi = \exp$ and $h = f_t$, it gives $\mathcal{W}p_t = p_t\big(\mathcal{W}f_t+\frac12\Gamma(f_t)\big)$, thus
\begin{align}
    \partial_tf_t = \mathcal{W}f_t+\frac12\,\Gamma(f_t) .
    \label{eq: step B}
\end{align}
We now substitute \eqref{eq: step C} and \eqref{eq: step B} into \eqref{eq: no drift}. Since $X$ does not depend on time, $\partial_t(Xf_t) = X(\partial_tf_t)$, and the left-hand side of \eqref{eq: no drift} can be written as
\begin{align*}
    \big(\tilde{\mathcal{W}}_t-\partial_t\big)(Xf_t) = \Big[\mathcal{W}(Xf_t)-X(\mathcal{W}f_t)\Big]+\Big[\Gamma(f_t,Xf_t)-\frac12X\big(\Gamma(f_t)\big)\Big] .
\end{align*}
Each bracket measures whether differentiating along $X$ commutes with the forward dynamics, the first one for the generator and the second one for the carr\'e du champ. Both vanish if the rotation generated by $\sigma_k$ is a symmetry of the forward dynamics, in the sense that for every smooth $h$
\begin{align}
    \mathcal{W}(Xh) = X(\mathcal{W}h),\qquad X\big(\Gamma(h)\big) = 2\,\Gamma(h,Xh) .
    \label{eq: step A}
\end{align}

These identities come from the unitary invariance of Gisin's equation for the depolarizing channel. For a unitary $U$, the map $M\mapsto UMU^\dagger$ preserves Hermiticity, tracelessness and the Hilbert-Schmidt inner product, so $U\sigma_kU^\dagger = \sum_jO_{jk}\sigma_j$ for a real orthogonal matrix $O$. If $\psi_t$ solves Gisin's equation with the jump operators $L_k$, then $U\psi_tU^\dagger$ solves it with the rotated jump operators $UL_kU^\dagger = \sum_jO_{jk}L_j$. This mixing does not change the law of the process: the drift is a sum of terms quadratic in the $L_k$, which is invariant, and the mixed noises $\sum_kO_{jk}d\xi_k$ are again independent complex Wiener processes. 

The map $\Phi_U(\psi) := U\psi U^\dagger$ thus intertwines the generator, $\mathcal{W}(h\circ\Phi_U) = (\mathcal{W}h)\circ\Phi_U$, and since $\Gamma(h) = \mathcal{W}(h^2)-2h\,\mathcal{W}h$ by the Leibniz rule above, it also intertwines the carr\'e du champ, $\Gamma(h\circ\Phi_U) = \Gamma(h)\circ\Phi_U$, for every smooth $h$. We differentiate both identities along $U_\epsilon = e^{i\epsilon\sigma_k}$ at $\epsilon=0$. The velocity of $\Phi_{U_\epsilon}$ is $X$, so that $\frac{d}{d\epsilon}(h\circ\Phi_{U_\epsilon})|_{\epsilon=0} = Xh$, and since $\Gamma$ is bilinear and symmetric, this gives \eqref{eq: step A}. Therefore \eqref{eq: no drift} holds, and what remains of $\eta_{k,T-\tau}(\psi_\tau)$ is a stochastic integral, which is a local martingale. It is a true martingale on $[0,T-\delta]$ with $\delta>0$, or on $[0,T]$ for smooth positive $p_0$, since there $\log p_t$ and its derivatives are bounded \cite{revuz1999continuous}.

We now turn to the Fisher information. Its second expression follows from \eqref{eq: tight frame} and \eqref{eq: D grad}: differentiating $f$ along both sides gives $\Gamma(f) = \gamma\sum_k(X_{i\sigma_k}f)^2$, and taking $f = \log p_t$ and the expectation under $p_t$ gives the second expression. This expression also shows why $I$ is monotone: it is a sum of second moments of the martingales. Indeed, $\psi_\tau$ has the law $p_{T-\tau}$, so that $I(T-\tau) = \gamma\sum_k\mathbb{E}[M_k(\tau)^2]$ with $M_k(\tau) := \eta_{k,T-\tau}(\psi_\tau)$. By the orthogonality of martingale increments, $\mathbb{E}[M_k(\tau)^2]$ is non-decreasing in $\tau$, and $I$ is therefore non-increasing in $t$.

Finally, we need the upper bound $I(t)\le2(d-1)/t$. Written in Riemannian terms, $I$ is the usual Fisher information of $p_t$: by \eqref{eq: depol brownian} and \eqref{eq: D grad}, $\mathcal{W} = \gamma\Delta_g$ and $D = 2\gamma g^{-1}$, hence $\Gamma(f) = 2\gamma|\nabla f|_g^2$ and $I(t) = 2\gamma\int|\nabla\log p_t|_g^2\,p_t\,dm$. Along the heat flow, such a quantity is controlled by the Li-Yau inequality \cite{li1986parabolic}: every positive solution of the heat equation $\partial_su = \Delta u$ on a compact manifold of dimension $n$ with non-negative Ricci curvature satisfies
\begin{align*}
    \frac{|\nabla u|_g^2}{u^2}-\frac{\partial_su}{u}\le\frac{n}{2s} .
\end{align*}
This applies here. The metric $g$ is a constant multiple of the Fubini-Study metric, which is an Einstein metric with positive Ricci curvature \cite{besse1987einstein}, on a manifold of dimension $n = 2(d-1)$, and $u(s) := p_{s/\gamma}$ solves $\partial_su = \Delta u$ because $\mathcal{W} = \gamma\Delta_g$. Multiplying the inequality by $u$, integrating over $\mathcal{P}$ and using $\int\partial_su\,dm = \frac{d}{ds}\int u\,dm = 0$, we find $\int|\nabla\log u|_g^2\,u\,dm\le n/(2s)$. With $s = \gamma t$, this is $I(t)\le2\gamma\,n/(2\gamma t) = 2(d-1)/t$.
\end{proof}

\section{Proof of Proposition \ref{prop: discretization}}\label{app: discretization}

\begin{proof}
We have to bound the relative entropy between the law $p_\delta$ of the exact backward process at reverse time $T-\delta$ and the law $\hat p_\delta$ of the output of the scheme. Appendix \ref{app: girsanov} bounds such a relative entropy for two continuous-time processes that have the same noise and whose drifts differ along the noise fields. The scheme proceeds in discrete steps, so we first write its output as the endpoint of a continuous process. In the $n$-th step, the scheme freezes the Hamiltonian at $H_n = \gamma\sum_ka_k(\hat\psi_n,t_n)\sigma_k$, then rotates and diffuses. Since the splitting of the frozen generator into a rotation and a diffusion is exact, this is the same as solving \eqref{eq: depol hamiltonian} on $[\tau_n,\tau_{n+1})$ with $H_t(\psi_\tau)$ replaced by $H_n$. Hence $\hat p_\delta$ is the law at reverse time $T-\delta$ of the process started from $q$ that solves \eqref{eq: depol hamiltonian} with the frozen Hamiltonian. This process has the same noise as the exact backward process, and the two drifts differ by the tangent vector
\begin{align*}
    \Delta_\tau := i\big[H_{T-\tau}(\psi_\tau)-H_n,\psi_\tau\big] = \gamma\sum_ke_k\,X_{i\sigma_k}(\psi_\tau),\qquad e_k := \eta_{k,T-\tau}(\psi_\tau)-a_k(\psi_{\tau_n},t_n).
\end{align*}

The argument of Appendix \ref{app: girsanov} used only that the drift difference lies along the noise fields. For every decomposition of $\Delta_\tau$ along the noise fields, with bounded adapted coefficients $\theta$, it gives
\begin{align*}
    \Delta_\tau = \sum_{k,a}\theta_{ka}V_{ka}(\psi_\tau)\quad\Longrightarrow\quad\mathrm{KL}\big(p_\delta\,\big\|\,\hat p_\delta\big) \le \mathrm{KL}\big(p_T\,\big\|\,q\big) + \frac12\,\mathbb{E}\int_0^{T-\delta}\sum_{k,a}\theta_{ka}^2\,d\tau .
\end{align*}
The decomposition is not unique, because the $2(d^2-1)$ real noise fields span a tangent space of dimension only $2(d-1)$, and the choice dictates the bound. The obvious choice is based on the relation $X_{i\sigma_k} = \sqrt{2/\gamma}\,V_{k2}$. It sets $\theta_{k1}=0$ and $\theta_{k2} = \sqrt{2\gamma}\,e_k$, so that $\frac12\sum_{k,a}\theta_{ka}^2 = \gamma\sum_ke_k^2$, which is twice too large. We therefore look for the decomposition with the smallest norm. To find it, we fix $\psi = \psi_\tau$ and introduce the linear maps
\begin{align*}
    S:\ \mathbb{R}^{2(d^2-1)}\to T_\psi\mathcal{P},\quad \theta\mapsto\sum_{k,a}\theta_{ka}V_{ka}(\psi),\qquad \mathcal{R}:\ \mathbb{R}^{d^2-1}\to T_\psi\mathcal{P},\quad e\mapsto\sum_ke_kX_{i\sigma_k}(\psi),
\end{align*}
from the noise coefficients and from the rotation coefficients to tangent vectors, so that we are looking for the smallest solution of $S\theta = \Delta_\tau = \gamma\mathcal{R}e$. Both maps are related to the diffusion matrix: $D = SS^T$, since $D^{ij} = \sum_aV_a^iV_a^j$, and $D = \gamma\mathcal{R}\mathcal{R}^T$ by \eqref{eq: tight frame} and \eqref{eq: D grad}. Moreover, $D$ is invertible, with $D^{-1} = g/2\gamma$ by \eqref{eq: D grad}. The solution with the smallest norm is then $\theta := S^TD^{-1}\Delta_\tau$, and it satisfies
\begin{align*}
    \sum_{k,a}\theta_{ka}^2 = \big\langle D^{-1}\Delta_\tau,\,SS^TD^{-1}\Delta_\tau\big\rangle = \big\langle D^{-1}\Delta_\tau,\,\Delta_\tau\big\rangle = \gamma^2\,e^T\mathcal{R}^T\big(\gamma\mathcal{R}\mathcal{R}^T\big)^{-1}\mathcal{R}e = \gamma\,e^T\Pi_{\mathcal{R}}\,e\le\gamma\sum_ke_k^2 ,
\end{align*}
where $\Pi_{\mathcal{R}} = \mathcal{R}^T(\mathcal{R}\mathcal{R}^T)^{-1}\mathcal{R}$ is the orthogonal projection onto the range of $\mathcal{R}^T$. These coefficients are adapted, and they are bounded because $D^{-1}$ is smooth on the compact manifold $\mathcal{P}$ and $e$ is bounded. Therefore
\begin{align*}
    \mathrm{KL}\big(p_\delta\,\big\|\,\hat p_\delta\big) \le \mathrm{KL}\big(p_T\,\big\|\,q\big) + \frac\gamma2\sum_n\int_{\tau_n}^{\tau_{n+1}}\mathbb{E}\Big[\sum_k\big(\eta_{k,T-\tau}(\psi_\tau)-a_k(\psi_{\tau_n},t_n)\big)^2\Big]d\tau ,
\end{align*}
where the expectation is over the exact backward process.

It remains to bound the last term by the last two terms of \eqref{eq: discretization bound}. The difference $\eta_{k,T-\tau}(\psi_\tau)-a_k(\psi_{\tau_n},t_n)$ combines two effects: at the beginning of the step the coefficient $a_k$ differs from the exact one, and within the step the exact coefficient moves away from its initial value. We separate them,
\begin{align*}
    \eta_{k,T-\tau}(\psi_\tau)-a_k(\psi_{\tau_n},t_n) = \big[\eta_{k,t_n}(\psi_{\tau_n})-a_k(\psi_{\tau_n},t_n)\big]+\big[\eta_{k,T-\tau}(\psi_\tau)-\eta_{k,t_n}(\psi_{\tau_n})\big],
\end{align*}
and use $(x+y)^2\le2x^2+2y^2$. The first effect is the score error. Since $\psi_{\tau_n}$ has the law $p_{t_n}$, it gives the second term of \eqref{eq: discretization bound}. The second effect is where the martingale property enters. By part 1 of Lemma \ref{lem: depol structure}, $M_k(\tau) := \eta_{k,T-\tau}(\psi_\tau)$ is a martingale, so its increments are orthogonal to its past and
\begin{align*}
    \mathbb{E}\big[(M_k(\tau)-M_k(\tau_n))^2\big] = \mathbb{E}\big[M_k(\tau)^2\big]-\mathbb{E}\big[M_k(\tau_n)^2\big] .
\end{align*}
Summed over $k$, these second moments are the Fisher information, $\gamma\sum_k\mathbb{E}[M_k(\tau)^2] = I(T-\tau)$, and $I(T-\tau)\le I(t_{n+1})$ for $\tau\le\tau_{n+1}$, since $I$ is non-increasing. The contribution of the $n$-th step is therefore at most $h_n[I(t_{n+1})-I(t_n)]$, which gives the last term of \eqref{eq: discretization bound}. With exact coefficients the first effect is absent, no splitting is needed, and the last term keeps its factor $\frac12$.
\end{proof}
\section{Proof of Corollary \ref{cor: step sizes}}\label{app: step sizes}
The condition $h_n\le\kappa t_{n+1}$ allows each step to be a fixed fraction of the time that remains. It turns the discretization term into $\kappa$ times a sum that can be bounded by a summation by parts and by the Li-Yau bound of Lemma \ref{lem: depol structure}.

\begin{proof}
Since $I$ is non-increasing, every term of the sum is non-negative, and the condition $h_n\le\kappa t_{n+1}$ bounds the sum by $\kappa\sum_nt_{n+1}[I(t_{n+1})-I(t_n)]$. A summation by parts gives
\begin{align*}
    \sum_{n=0}^{N-1}t_{n+1}\big[I(t_{n+1})-I(t_n)\big] = \sum_{n=0}^{N-1}\big[t_{n+1}I(t_{n+1})-t_nI(t_n)\big] + \sum_{n=0}^{N-1}h_n\,I(t_n) = \delta I(\delta)-TI(T)+\sum_{n=0}^{N-1}h_n\,I(t_n),
\end{align*}
where the first sum telescopes. Since $I$ is non-increasing, $I(t_n)\le I(t)$ for $t\in[t_{n+1},t_n]$, so the last sum is at most $\int_\delta^TI(t)\,dt$. Dropping the negative term $-TI(T)$ and using the bound $I(t)\le2(d-1)/t$ of Lemma \ref{lem: depol structure}, both for $\delta I(\delta)$ and in the integral, we obtain
\begin{align*}
    \sum_{n=0}^{N-1}t_{n+1}\big[I(t_{n+1})-I(t_n)\big]\le\delta I(\delta)+\int_\delta^TI(t)\,dt\le2(d-1)\Big(1+\log\frac T\delta\Big).
\end{align*}
Finally, we check that the geometric grid satisfies the condition. Away from the cut-off, $t_n = T(1+\kappa)^{-n}$ and $h_n = \kappa t_{n+1}$ exactly. Only the last step is affected by the cut-off at $\delta$: the choice of $N$ gives $T(1+\kappa)^{-N}\le\delta<T(1+\kappa)^{-(N-1)}$, so the last step runs from $t_{N-1}\le(1+\kappa)\delta$ to $t_N = \delta$ and has $h_{N-1}\le\kappa\delta$.
\end{proof}

\section{Details of the numerical tests}\label{app: numerics}
All the initial ensembles of Section \ref{sec: examples} are invariant under the unitaries that fix a pure state $\psi_0$. For such ensembles, every quantity of interest depends on the state only through its fidelity with $\psi_0$, and the dynamics reduces to a diffusion of a single variable, whose spectral decomposition is explicit. In this appendix we describe this reduction, show how the discrete scheme and the perturbed scores act on the reduced variable, and summarize the implementation.

\subsection{Reduction to one variable}\label{app: numerics reduction}
Let $x = 2\,\mathrm{tr}(\psi_0\psi)-1$ as in \S\ref{subsec: setting}. As noted after Proposition \ref{prop: girsanov}, the unitaries that fix $\psi_0$ act transitively on the level sets of $x$, and the depolarizing dynamics is unitarily covariant (Appendix \ref{app: depol structure}). If $p_0$ is invariant under these unitaries, so is $p_t$, which is then a function of $x$ alone. Moreover, $x$ is itself a diffusion. By \eqref{eq: observable sde} and \eqref{eq: depol d}, its drift is $\mathcal{W}x = 2\langle\mathcal{L}^\dagger[\psi_0]\rangle = 2\gamma(2-d-dx)$, and its diffusion coefficient is $2\sum_k|\partial_{B_k}x|^2 = 4\gamma(1-x^2)$, where $\partial_{B_k}x = 2\langle\psi|\psi_0(L_k-l_k)|\psi\rangle$ and the sum is evaluated with the completeness relation \eqref{eq: completeness}. On functions of $x$, the generator therefore acts as
\begin{align}
    \mathcal{W}_x = 2\gamma\Big[(1-x^2)\,\partial_x^2 + (2-d-dx)\,\partial_x\Big] = \frac{2\gamma}{(1-x)^{d-2}}\,\partial_x\Big((1-x)^{d-2}(1-x^2)\,\partial_x\Big).
    \label{eq: fidelity generator}
\end{align}
The second form shows that $\mathcal{W}_x$ is symmetric with respect to $\nu(dx) = \frac{d-1}{2}\big(\frac{1-x}{2}\big)^{d-2}dx$, the law of $x$ under the Fubini-Study measure. Its eigenfunctions are the Jacobi polynomials $\phi_k(x) = P_k^{(d-2,0)}(x)$, which are orthogonal for this weight, with $\int\phi_k^2\,d\nu = (d-1)/(2k+d-1)$ \cite{dlmf}, and its eigenvalues are the $-\lambda_k$ of \eqref{eq: depol spectrum}. The density of an invariant ensemble therefore evolves as
\begin{align}
    p_t(\psi) = \sum_{k=0}^\infty\frac{2k+d-1}{d-1}\,e^{-\lambda_kt}\,\mathbb{E}_{p_0}\big[\phi_k\big]\,\phi_k(x),
    \label{eq: jacobi heat kernel}
\end{align}
with $\mathbb{E}_{p_0}[\phi_k] = \phi_k(1) = \binom{k+d-2}{k}$ for the Dirac mass at $\psi_0$. In this basis, the transition operator $P_s$ simply multiplies the coefficient of $\phi_k$ by $e^{-\lambda_ks}$, which gives the exact output \eqref{eq: h transform} of the backward SDE started from any invariant $q$. For $d=2$, $\nu$ is the uniform law $dx/2$ and the $\phi_k$ are the Legendre polynomials. Finally, given $x$, the state is uniformly distributed on the level set, for $p_t$ as well as for all the outputs we consider. Relative entropies and total variation distances between laws of states therefore equal those between the corresponding laws of $x$.

\subsection{Score, Fisher information and the discrete scheme}\label{app: numerics scheme}
By the chain rule, the scores are $\partial_{B_k}\log p_t = g_t\,\partial_{B_k}x$, with $\sum_k|\partial_{B_k}x|^2 = 2\gamma(1-x^2)$, so that the Fisher information of Lemma \ref{lem: depol structure} is $I(t) = \mathbb{E}_{p_t}\big[4\gamma(1-x^2)\,g_t^2\big]$. The scheme of Definition \ref{def: discrete time scheme} preserves the invariance, so it can also be run on the law of $x$. With $\eta_{k,t} = g_t\,X_{i\sigma_k}x$ and $X_{i\sigma_k}x = 2\,\mathrm{tr}\big(\sigma_k\,i[\psi,\psi_0]\big)$, the second identity of \eqref{eq: completeness consequences} gives the Hamiltonian $H_t = 4\gamma\,g_t\,i[\psi,\psi_0]$, which acts in the plane spanned by $|\psi_0\rangle$ and $|\psi\rangle$. Writing $|\psi\rangle = \cos\frac\theta2|\psi_0\rangle + \sin\frac\theta2|\chi\rangle$ up to a phase, with $\langle\psi_0|\chi\rangle = 0$ and $x = \cos\theta$, the rotation step leaves $|\chi\rangle$ unchanged and maps $\theta$ to $\theta - 4\gamma h_n\,g_{t_n}(x)\sin\theta$. The diffusion step multiplies the coefficient of $\phi_k$ by $e^{-\lambda_kh_n}$. The law of $x$ is thus propagated step by step, by a change of variables followed by a diagonal multiplication in the basis of Jacobi polynomials. For the perturbed scores, $g_{t_n}$ is replaced by $g^s_{t_n}$ in the rotation step.

\subsection{Perturbed scores}\label{app: numerics perturbed}
With the learned scores $s_k = g^s_t\,\partial_{B_k}x$, the backward generator on functions of $x$ follows from \eqref{eq: step C} with the score replaced, and it reads $\mathcal{W}_xh + 4\gamma(1-x^2)\,g^s_t\,\partial_xh$. Its adjoint with respect to $\nu$ gives the equation for the density $q_\tau$ of $x$ with respect to $\nu$, at physical time $t = T-\tau$,
\begin{align*}
    \partial_\tau q_\tau = -\frac1n\,\partial_x\Big[n\,(1-x^2)\big(4\gamma\,g^s_t\,q_\tau - 2\gamma\,\partial_xq_\tau\big)\Big],\qquad n := \frac{d\nu}{dx},
\end{align*}
which for $g^s_t = g_t$ is solved by $q_\tau = p_{T-\tau}$. We integrate this equation from $q_0 = 1$ with finite volumes on $4000$ uniform cells, central fluxes and the Crank-Nicolson scheme with time step $5\times10^{-4}$. As a check, with the exact score it reproduces the relative entropy of the exact output \eqref{eq: h transform} to within $0.2\%$ for $d\le6$, and to within $4\%$ for $d=8$, where this relative entropy is only $2.3\times10^{-6}$.

\subsection{Implementation}
Densities are represented by their coefficients on the Jacobi polynomials $\phi_k$, $k\le200$, and evaluated on $3000$ Gauss-Legendre nodes in $x$, whose weights are multiplied by the density of $\nu$. The exact outputs are normalized to within $10^{-8}$. For the Dirac ensemble at small times, the density becomes extremely small far from $\psi_0$; where it falls below $10^{-13}$ of its maximum, we continue the score by its last reliable value, and these regions carry no weight. The smooth ensemble is the mixture
\begin{align}
    p_0(x) = 0.4\,\frac{e^{15(x-1)}}{\int e^{15(x-1)}\,d\nu} + 0.6\,\frac{e^{-(x+0.3)^2/(2\cdot0.08^2)}}{\int e^{-(x+0.3)^2/(2\cdot0.08^2)}\,d\nu}
    \label{eq: smooth p0}
\end{align}
of a cap around $\psi_0$, which carries $40\%$ of the weight, and a ring around $x = -0.3$. In higher dimension the volume factor of $\nu$ moves the first component away from $\psi_0$: its law of $x$ peaks at $x = 1-(d-2)/15$. The Fisher information of the smooth ensemble at $t=0$ is $I(0) = 94.7$, $104$, $114$, $131$ and $147$ for $d = 2,3,4,6,8$.

\section{Relation to the reverse SDE of Bompais, Gu\c{t}\u{a} and Garrahan}\label{app: bgg}
The reverse SDE of \cite{garrahan2026} is derived from the Kolmogorov equations in the real vector space $\mathrm{Herm}_d$ of Hermitian matrices, with the inner product $\mathrm{tr}(AB)$, in which the pure states form a submanifold of measure zero. In this appendix we show that it follows rigorously from the Stratonovich reversal of Section \ref{sec: reversal}, once its divergence term is read as a distribution. Unlike Gisin's equation, the forward process of \cite{garrahan2026} has one real Wiener process per jump operator, which describes homodyne detection,
\begin{align}
    d\psi = \mathcal{L}(\psi)\,dt + \sum_mK_m(\psi)\,dW_m,\qquad K_m := X_{L_m},
    \label{eq: homodyne sse}
\end{align}
with the Lindbladian \eqref{eq: lindblad}, the fields \eqref{eq: X_M} and independent real Wiener processes $W_m$. Let $\mathcal{D} := \mathcal{L}+i[H,\,\cdot\,]$ be the dissipator, $D := \sum_mK_m\otimes K_m$ the diffusion tensor, and $\mu_t$ the law of $\psi_t$, regarded as a probability measure on $\mathrm{Herm}_d$. It is concentrated on $\mathcal{P}$, in the sense that $\int\varphi\,d\mu_t = \int_{\mathcal{P}}\varphi\,p_t\,dm$ for functions $\varphi$ on $\mathrm{Herm}_d$, where $p_t$ is the density of $\psi_t$ with respect to $dm$. In reverse time $\tau = T-t$, the reverse SDE of \cite[Theorem~1]{garrahan2026} reads
\begin{align}
    d\psi = \Big(-i\big[H+H^{\rm fb}_\tau,\psi\big]+\mathcal{D}(\psi)\Big)d\tau + \sum_mK_m(\psi)\,d\tilde W_m,\qquad H^{\rm fb}_\tau := -2H+i[\Xi_\tau,\psi],
    \label{eq: bgg reversal}
\end{align}
with the Hermitian matrix
\begin{align}
    \Xi_\tau := -2\,\mathcal{D}(\psi)+\frac{1}{\mu_t}\,\mathrm{Div}\big(D\mu_t\big) .
    \label{eq: bgg Xi}
\end{align}
Here $\mathrm{Div}$ denotes the divergence over the matrix entries of $\psi$, $(\mathrm{Div}\,D)_{ij} := \sum_{h,k}\partial D_{ij,hk}/\partial\psi_{hk}$, while $\mathrm{div}$ is reserved  for the divergence on $\mathcal{P}$ with respect to $dm$.  Since $\mu_t$ has no density with respect to the Lebesgue measure of $\mathrm{Herm}_d$, this term is understood in a distributional sense. For smooth test functions $\varphi$ on $\mathrm{Herm}_d$, it is defined by  
\begin{align}
    \langle\mathrm{Div}(D\mu_t),\varphi\rangle := -\int D\nabla\varphi\,d\mu_t.
\end{align}
In this expression $D\nabla \varphi := \sum_m(K_m\varphi)\,K_m$ in the notation of \cite{garrahan2026}. On $\mathcal{P}$ it depends only on the restriction of $\varphi$ to $\mathcal{P}$, because the $K_m$ are tangent, so it coincides with $D\nabla(\varphi|_{\mathcal{P}})$ as defined in Section \ref{sec: reversal}.

\begin{proposition}[Relation to \cite{garrahan2026}]\label{prop: bgg}
Assume that $p_t$ is smooth and positive for $t>0$; the matrix-valued distribution $\mathrm{Div}(D\mu_t)$ has a density with respect to $\mu_t$, namely
    \begin{align}
        \frac{1}{\mu_t}\,\mathrm{Div}\big(D\mu_t\big) = \sum_m\Big[\mathrm{div}_{p_t}(K_m)\,K_m+(K_m\cdot\partial)K_m\Big].
        \label{eq: bgg divergence}
    \end{align}
As a result, \eqref{eq: bgg reversal} coincides with the reverse SDE of Section \ref{sec: reversal} for \eqref{eq: homodyne sse}.
\end{proposition}

\begin{proof}
By definition of $p_t$ 
\begin{align}
-\int D\nabla\varphi\,d\mu_t = -\sum_m\int_{\mathcal{P}}(K_m\varphi)K_m\,p_t\,dm.
\end{align}
The $K_m$ are tangent, so this involves only the derivatives of $\varphi$ along $\mathcal{P}$, and we can integrate by parts on $\mathcal{P}$. For a function $f$ on $\mathcal{P}$, \eqref{eq: div p} gives
\begin{align}
    -\int(K_m\varphi)\,f\,p_t\,dm = \int\varphi\,\mathrm{div}\big(p_tfK_m\big)\,dm = \int\varphi\,\big(f\,\mathrm{div}_{p_t}K_m+K_mf\big)\,p_t\,dm .
\end{align}
Taking for $f$ the entries of $K_m$, whose derivatives $K_mf$ are the entries of $(K_m\cdot\partial)K_m$, and summing over $m$ gives
\begin{align}
    \langle\mathrm{Div}(D\mu_t),\varphi\rangle &= \int_\mathcal{P}\varphi \Big(\sum_m\Big[\mathrm{div}_{p_t}(K_m)\,K_m+(K_m\cdot\partial)K_m\Big]\Big)p_t\,dm\nonumber\\
    &= \int \varphi \Big(\sum_m\Big[\mathrm{div}_{p_t}(K_m)\,K_m+(K_m\cdot\partial)K_m\Big]\Big)\,d\mu_t.
\end{align}
Since this holds for all test functions $\varphi$, we get \eqref{eq: bgg divergence}. 

With this result we can relate \eqref{eq: bgg reversal} with our reverse SDE obtained using the machinery in Section \ref{sec: reversal}. The reversal \eqref{eq: stratonovich reversal} holds for the real noise fields $K_m$, and gives the reverse Stratonovich drift $-V_0+\sum_m\mathrm{div}_{p_t}(K_m)K_m$ while the It\^o correction is $\frac12\sum_m(K_m\cdot\partial)K_m$. The reverse It\^o drift for the forward process \eqref{eq: homodyne sse} is thus given by
\begin{align*}
    \tilde b = -\mathcal{L}(\psi)+\sum_m\Big[\mathrm{div}_{p_t}(K_m)\,K_m+(K_m\cdot\partial)K_m\Big] = -\mathcal{L}(\psi)+\frac{1}{\mu_t}\,\mathrm{Div}\big(D\mu_t\big) = i[H,\psi]+\mathcal{D}(\psi)+\Xi_\tau .
\end{align*}
To show that this coincides with the drift of \eqref{eq: bgg reversal}, we just need to show that $\Xi_\tau = [[\Xi_\tau,\psi],\psi]$. By \S\ref{subsec: pure state geometry}, the tangent space at $\psi$ is $T_\psi = \lbrace u\in\mathrm{Herm}_d : u = \psi u+u\psi\rbrace$, and multiplying this relation by $\psi$ on both sides shows that $\psi u\psi = 0$ for $u\in T_\psi$. The orthogonal projection onto $T_\psi$ is $\Pi_\psi(Y) := \psi Y+Y\psi-2\,\psi Y\psi$, and $N_\psi := 1-\Pi_\psi$ is the projection onto the normal space. Expanding the commutators gives for every Hermitian $Y$
\begin{align}
    \big[[Y,\psi],\psi\big] = \Pi_\psi(Y) .
\end{align}
Therefore, we just need to show that $\Xi_\tau$ is a tangent vector. Since $\mathrm{div}_{p_t}(K_m)K_m$ are tangent, 
\begin{align}
    N_\psi(\Xi_\tau) = -2 N_\psi\big(\mathcal{D}(\psi)\big) + \sum_mN_\psi\big((K_m\cdot\partial)K_m\big).
\end{align}
Since the Hamiltonian part $-i[H,\psi] = X_{-iH}(\psi)$ is tangent, $N_\psi\big(\mathcal{D}(\psi)\big) = N_\psi\big(\mathcal{L}(\psi)\big)$. Moreover, the Stratonovich drift $V_0 = \mathcal{L}(\psi)-\frac12\sum_m(K_m\cdot\partial)K_m$ of \eqref{eq: homodyne sse} is tangent as well, so $N_\psi\big(\mathcal{L}(\psi)\big) = \frac12\sum_mN_\psi\big((K_m\cdot\partial)K_m\big)$. This concludes the proof.
\end{proof}

\end{document}